\documentclass[oneside,english,reqno,11pt]{amsart}

\calclayout
\usepackage{geometry}
\usepackage{babel}
\usepackage{enumitem}
\usepackage{amssymb,amsmath,amsthm,amsfonts,latexsym,esint,dsfont}
\usepackage{bbm}
\usepackage{mathrsfs}
\usepackage{appendix}

\usepackage{hyperref}
\usepackage{color}
\usepackage{datetime}
\makeatletter
\numberwithin{equation}{section}
\numberwithin{figure}{section}
\theoremstyle{plain}
\newtheorem{theorem}{Theorem}[section]
\newtheorem{proposition}[theorem]{Proposition}

\theoremstyle{remark}
\newtheorem{remark}[theorem]{\protect\remarkname}
\theoremstyle{plain}
\newtheorem{lemma}[theorem]{\protect\lemmaname}
\newlist{casenv}{enumerate}{4}
\setlist[casenv]{leftmargin=*,align=left,widest={iiii}}
\setlist[casenv,1]{label={{\itshape\ \casename} \arabic*.},ref=\arabic*}
\setlist[casenv,2]{label={{\itshape\ \casename} \roman*.},ref=\roman*}
\setlist[casenv,3]{label={{\itshape\ \casename\ \alph*.}},ref=\alph*}
\setlist[casenv,4]{label={{\itshape\ \casename} \arabic*.},ref=\arabic*}

\makeatletter
\def\@tocline#1#2#3#4#5#6#7{\relax
  \ifnum #1>\c@tocdepth 
  \else
    \par \addpenalty\@secpenalty\addvspace{#2}%
    \begingroup \hyphenpenalty\@M
    \@ifempty{#4}{%
      \@tempdima\csname r@tocindent\number#1\endcsname\relax
    }{%
      \@tempdima#4\relax
    }%
    \parindent\z@ \leftskip#3\relax \advance\leftskip\@tempdima\relax
    \rightskip\@pnumwidth plus4em \parfillskip-\@pnumwidth
    #5\leavevmode\hskip-\@tempdima
      \ifcase #1
       \or\or \hskip 1em \or \hskip 2em \else \hskip 3em \fi%
      #6\nobreak\relax
      \dotfill
      \hbox to\@pnumwidth{\@tocpagenum{#7}}
    \par
    \nobreak
    \endgroup
  \fi}
\makeatother

\makeatletter
\def\th@plain{%
	\thm@notefont{}
	\itshape 
}
\def\th@definition{%
	\thm@notefont{}
	\normalfont 
}
\makeatother

\providecommand{\corollaryname}{Corollary}
\providecommand{\lemmaname}{Lemma}
\providecommand{\remarkname}{Remark}
\providecommand{\casename}{Case}

\newcommand{\bx}{\mathbf{x}}
\newcommand\norm[1]{\left\lVert#1\right\rVert}
\newcommand{\LLL}{\mathcal{LLL}}
\newcommand{\LTF}{L_{\Omega}^{\rm TF}}
\newcommand{\ETF}{E_{\Omega}^{\rm TF}}
\newcommand{\uGP}{u^{\rm GP}}
\newcommand{\cEGP}{\mathcal{E}_\Omega^{\rm GP}}
\newcommand{\EGP}{E_\Omega^{\rm GP}}

\newcommand{\cELLL}{\mathcal{E}_\Omega^{\rm LLL}}
\newcommand{\ELLL}{E_\Omega^{\rm LLL}}

\DeclareMathOperator{\R}{\mathbb{R}}
\newcommand{\eps}{\varepsilon}

\newcommand{\cGGP}{\mathcal{F}^{\rm GP}}
\newcommand{\GGP}{F ^{\rm GP}}

\newcommand{\cGLLL}{\mathcal{F} ^{\rm LLL}}
\newcommand{\GLLL}{F^{\rm LLL}}

\newcommand{\cGTF}{\mathcal{F} ^{\rm TF}}
\newcommand{\GTF}{F ^{\rm TF}}
\newcommand{\rhoTF}{\varrho ^{\rm TF}}

\allowdisplaybreaks
\begin{document}
	
	\title[TF profile of a fast rotating BEC]{Thomas--Fermi profile of a fast rotating \\ Bose--Einstein condensate}
	
	\author[D.-T. NGUYEN]{Dinh-Thi NGUYEN}
	\address[Dinh-Thi NGUYEN]{Ecole Normale Sup\'erieure de Lyon \& CNRS, UMPA (UMR 5669)} 
	\email{\href{dinh.nguyen@ens-lyon.fr}{dinh.nguyen@ens-lyon.fr}}

	\author[N. ROUGERIE]{Nicolas ROUGERIE}
	\address[Nicolas ROUGERIE]{Ecole Normale Sup\'erieure de Lyon \& CNRS, UMPA (UMR 5669)} 
	\email{\href{nicolas.rougerie@ens-lyon.fr}{nicolas.rougerie@ens-lyon.fr}}
	
	
	\date{July 2022, with additions, June 2024}

	\begin{abstract}
		We study the minimizers of a magnetic 2D non-linear Schr\"odinger energy functional in a quadratic trapping potential, describing a rotating Bose--Einstein condensate. We derive an effective Thomas--Fermi-like model in the rapidly rotating limit where the centrifugal force compensates the confinement, and available states are restricted to the lowest Landau level. The coupling constant of the effective Thomas--Fermi functional is linked to the emergence of vortex lattices (the Abrikosov problem). We define it via a low density expansion of the energy of the corresponding homogeneous gas in the thermodynamic limit. 
	\end{abstract}
	
	\maketitle

	\tableofcontents
	
	\section{Introduction and main results}
	
	The remarkably versatile experimental conditions of cold atoms physics allow to emulate several condensed matter phenomena in a well-controlled fashion. A very interesting direction is to simulate the effect of an external magnetic field on a coherent matter wave, in analogy with the rich physics of superconductors, in particular of type II. Several experiments have observed quantized vortices in rotating Bose--Einstein condensates \cite{AboRamVogKet-01,MadChevWohDal-00,BreStoSeuDal-04,SchCodEngMogCor-04,CodHalEngSchTungCor-04}. In such systems, all the atoms of a Bose gas occupy the same quantum state, whence the phase coherence. The particles under consideration are neutral. Making them rotate allows to imitate the effect of a magnetic field by relying on the well-known analogy ``Coriolis force $\Leftrightarrow$ Lorentz force''.
	
	For a Bose--Einstein condensate in fast rotation, the centrifugal force spreads the gas in the plane perpendicular to the rotation axis. A 2D model is then appropriate \cite{AftBlaDal-05,CooKomRea-04,Ho-01} and the relevant energy is \cite{Aftalion-07,CorPinRouYng-11b}
	\begin{equation}\label{functional:GP-true}
	\mathcal{G}_{\Omega}^{\rm GP}[\psi] = \frac{1}{2}\int_{\mathbb{R}^{2}} \big|\big(-\mathrm{i}\nabla - \Omega\mathbf{x}^{\perp}\big)\psi\big|^{2} + G|\psi|^{4} + \left(1-\Omega^{2}\right)|\mathbf{x}|^{2}|\psi|^{2},
	\end{equation}
	where $\mathbf{x} = (x_{1},x_{2}) \in \mathbb{R}^{2}$, $\mathbf{x}^{\perp} = (-x_{2},x_{1})$, $G > 0$ measures repulsive interactions between the gas' atoms and $\Omega > 0$ is the rotational velocity. Note that we need $\Omega < 1$ in order for the energy to be bounded below. The \emph{rapidly rotating regime} corresponds to the limit $\Omega  \nearrow 1$.
	
	The model based on the Gross--Pitaevskii (GP) energy \eqref{functional:GP-true} is an approximation of the quantum mechanical many-body problem for $N$ bosons \cite{LieSeiSolYng-05,Rougerie-EMS,Golse-13,BenPorSch-15,Schlein-08}. The rigorous derivation was first performed by Lieb and Seiringer \cite{LieSei-06} in the case of fixed rotation and by Lieb, Seiringer, and Yngvason \cite{LieSeiYng-00} in the case of no rotation. Concerning the rapidly rotating regime, see \cite{LieSeiYng-09,LewSei-09,BruCorPicYng-08}. The above 2D-GP model was rigorously derived from 3D-GP by Aftalion and Blanc \cite{AftBla-06} in the limit $\Omega \nearrow 1$. 
	
	In order to study the asymptotics of the problem when $\Omega \nearrow 1$, it is more convenient to make the change of variables
	$$
	u(\mathbf{x}) = \frac{1}{\sqrt{\Omega}} \psi\left(\frac{\mathbf{x}}{\sqrt{\Omega}}\right).
	$$
	The Gross--Pitaevskii energy functional gets rescaled as $\mathcal{G}_{\Omega}^{\rm GP}[\psi] = \Omega\mathcal{E}_{\Omega}^{\rm GP}[u]$ where
	\begin{equation}\label{functional:GP}
	\mathcal{E}_{\Omega}^{\rm GP}[u] = \frac{1}{2}\int_{\mathbb{R}^{2}} \big|\big(-\mathrm{i}\nabla - \mathbf{x}^{\perp}\big)u\big|^{2} + G|u|^{4} + \frac{1-\Omega^{2}}{\Omega^{2}}|\mathbf{x}|^{2}|u|^{2}.
	\end{equation}
	The corresponding minimization problem is 
	\begin{equation}\label{energy:GP}
	E_{\Omega}^{\rm GP} = \inf\left\{\mathcal{E}_{\Omega}^{\rm GP}[u] : u\in H^{1}(\mathbb{R}^{2}), \int_{\mathbb{R}^{2}}|u|^{2} = 1\right\}.
	\end{equation}
	The first term of the energy functional in \eqref{functional:GP} is reminiscent of type II superconductors near the second critical field $H_{c_{2}}$. It is well-known (see \cite{LanLif-65,LuPan-99,RouYng-19}) that the eigenvalues of the operator 
	\begin{equation}\label{eq:Landau hamil}
	\frac{1}{2}\big(-\mathrm{i}\nabla - \mathbf{x}^{\perp}\big)^{2}	 
	\end{equation}
	are $2k+1$ $(k \in \mathbb{N})$. The first eigenspace is called the \emph{lowest Landau level} (LLL) as used in \cite{GirJac-84,AftBlaNie-06a}. It is of infinite dimension and is given by
	\begin{equation}\label{space:LLL}
	\mathcal{LLL} := \left\{u(\mathbf{x}) = f(z) e^{-\frac{|z|^{2}}{2}} : f \text { analytic (holomorphic)}\right\} \cap L^{2}(\mathbb{R}^{2}).
	\end{equation}
	Here we used complex coordinates $\mathbb{R}^{2} \ni \mathbf{x} = (x_{1},x_{2}) \leftrightarrow z = x_{1} + \mathrm{i}x_{2} \in \mathbb{C}$.
	For such a $u \in \mathcal{LLL}$, we find that $\mathcal{E}_{\Omega}^{\rm GP}[u]$ is equal to
	\begin{equation}\label{functional:LLL}
	\mathcal{E}_{\Omega}^{\rm LLL}[u] := 1 + \frac{1}{2}\int_{\mathbb{R}^{2}} G|u|^{4} + \frac{1-\Omega^{2}}{\Omega^{2}}|\mathbf{x}|^{2}|u|^{2}.
	\end{equation}
	In the fast rotating limit $\Omega \nearrow 1$, it is easy to see that the minimization of the last two terms yields a small quantity. Since the gap of the Landau operator~\eqref{eq:Landau hamil} is fixed, it makes sense to simplify the problem by projecting it in the ground eigenspace~\eqref{space:LLL}. This approximation has motivated numerous studies, e.g., \cite{GirJac-84,ButRok-99,Ho-01,AftBlaDal-05,AftBla-06,AftBlaNie-06a,AftBlaNie-06b,BlaRou-08,Rougerie-11}, and has been mathematically justified in \cite{AftBla-08}. The corresponding evolution equation \cite{Nier-07} also attracted attention recently \cite{GerGerTho-19,GerTho-16,GerThoSch-24,ThoSch-21}.
	
	%
	In \cite{AftBla-06}, it was proved that, as $\Omega$ tends to $1$, 
	\begin{equation}\label{energy:behavior-lll}
	E_{\Omega}^{\rm LLL} = 1 + \mathcal{O}\left(G^{\frac{1}{2}}(1-\Omega^{2}\big)^{\frac{1}{2}}\right),
	\end{equation}
	provided that $G(1-\Omega^2)^{-1} \to \infty$ as $\Omega \nearrow 1$. Here $E_{\Omega}^{\rm LLL}$ is the minimization of $\mathcal{E}_{\Omega}^{\rm GP}$, restricted to $\mathcal{LLL}$, i.e.,
	\begin{equation}\label{energy:LLL}
	E_{\Omega}^{\rm LLL} = \inf\left\{\mathcal{E}_{\Omega}^{\rm GP}[u] : u\in \mathcal{LLL}, \int_{\mathbb{R}^{2}}|u|^{2} = 1\right\}  = \inf\left\{\mathcal{E}_{\Omega}^{\rm LLL}[u] : u\in \mathcal{LLL}, \int_{\mathbb{R}^{2}}|u|^{2} = 1\right\}.
	\end{equation}
	Clearly, $E_{\Omega}^{\rm LLL}$ gives an upper bound to $E_{\Omega}^{\rm GP}$. For the energy lower bound, Aftalion and Blanc \cite{AftBla-08} proved that~\eqref{energy:GP} is well approximated by \eqref{energy:LLL} in the sense that, as $\Omega$ tends to $1$,
	\begin{equation}\label{energy:behavior-GP-LLL-fake}
	E_{\Omega}^{\rm GP} - E_{\Omega}^{\rm LLL} = o\left(G^{\frac{1}{2}}\left(1-\Omega^{2}\right)^{\frac{1}{2}}\right),
	\end{equation}
	for a fixed $G > 0$.
	
	The aim of this paper is to better characterize $E_{\Omega}^{\rm GP}$ as well as $E_{\Omega}^{\rm LLL}$ in the limit $\Omega \nearrow 1$. Without the constraint $u \in \mathcal{LLL}$, the minimization problem \eqref{energy:LLL} is exactly soluble and gives a density profile of the \emph{Thomas--Fermi} (TF) type. A conjecture made for example in \cite{AftBlaNie-06b} is as follows: in the limit $\Omega \nearrow 1$ the leading order effect of the constraint $u \in \mathcal{LLL}$ is to renormalize the interaction coefficient $G$. To calculate an approximation to the energy and matter density, one may thus solve a problem of Thomas--Fermi type. Mathematically, it is expected that the problem \eqref{energy:LLL} simplifies to
	\begin{equation}\label{energy:TF}
	E_{\Omega}^{\rm LLL} - 1 \underset{\Omega \nearrow 1}{\sim} E_{\Omega}^{\rm{TF}} := \inf\left\{\mathcal{E}_{\Omega}^{\rm{TF}}[\rho]: \rho \in L^{1}\cap L^{2}(\mathbb{R}^{2} ; \mathbb{R}^{+}), \int_{\mathbb{R}^{2}}\rho=1\right\},
	\end{equation}
	where 
	\begin{equation}\label{functional:TF}
	\mathcal{E}_{\Omega}^{\rm{TF}}[\rho] := \frac{1}{2}\int_{\mathbb{R}^{2}}e^{\rm Ab}(1)G\rho^{2} + \frac{1-\Omega^{2}}{\Omega^{2}}|\mathbf{x}|^{2}\rho.
	\end{equation}
	The parameter $e^{\rm Ab}(1)$ in \eqref{functional:TF} describes the contribution of a lattice of quantized vortices, which is related to the Abrikosov problem \cite{Almog-06,FouKac-11,AftSer-07,Abrikosov-57,KleRotAut-64,Sigal-13,SanSer-07} for a type II superconductor. In Section \ref{sec:homogeneous}, we will define $e^{\rm Ab}(1)$ using the thermodynamic energy per unit area at low density, i.e.,
	\begin{equation}\label{eq:def Ab}
	e^{\rm Ab}(1) = \frac{2}{G}\lim_{\varrho\to 0}\lim_{L\to\infty}\frac{E^{\rm GP}\left(K_L,\varrho L^{2}\right)/L^{2} - \varrho}{\varrho^{2}}.
	\end{equation}
	Here $K_L$ is the square 
	$$K_L = \left[-\frac{L}{2}, \frac{L}{2}\right]^{2}$$
	and $E^{\rm GP}(\mathcal{D},M)$ is the Neumann energy in the domain $\mathcal{D} \subset \mathbb{R}^{2}$ with mass $M > 0$, 
	\begin{equation}\label{energy:homogeneous-neumann}
	E^{\rm GP}(\mathcal{D},M) = \inf\left\{\mathcal{E}_{\mathcal{D}}^{\rm GP}[u] : u\in H^{1}(\mathcal{D}), \int_{\mathcal{D}}|u|^{2} = M\right\},
	\end{equation}
	where
	\begin{equation}\label{functional:homogeneous-neumann}
	\mathcal{E}_{\mathcal{D}}^{\rm GP}[u] = \frac{1}{2}\int_{\mathcal{D}} \big|\big(-\mathrm{i}\nabla - \mathbf{x}^{\perp}\big)u\big|^{2} + G|u|^{4}.
	\end{equation}
	A major open conjecture is that $e^{\rm Ab}(1)$ coincides with the Abrikosov value $\sim 1.1596$ obtained \cite{AftBla-06,AftBlaNie-06b} by using a $\mathcal{LLL}$ trial state with a hexagonal lattice of singly-quantized vortices with periodicity independent of the system size $L$ (i.e. fixed by the magnetic length only). This remains an open problem \cite{BlaLew-15}, linked to cristallization questions (Abrikosov lattices). Our (more modest) goal will be to prove that~\eqref{energy:TF} is true for the value of $e^{\rm Ab}(1)$ implicitly defined as in~\eqref{eq:def Ab} (see Theorem~\ref{thm:Abrikosov} below for more details on the definition). Our point is thus to justify rigorously a certain local density approximation (LDA). We are particularly interested in proving that the density profile of the full GP/LLL model is of Thomas--Fermi type when $\Omega \nearrow 1$, for this can be interpreted as a signature of vortex lattice inhomogeneities \cite{AftBlaDal-05}.
	
	Let $\rho_{\Omega}^{\rm{TF}}$ be the (unique) minimizer for $E_{\Omega}^{\rm{TF}}$ in \eqref{energy:TF}. By scaling 
	\begin{equation}\label{TF-min:scaled}
	\rho_{\Omega}^{\rm{TF}} (\mathbf{x}) = G^{-\frac{1}{2}}\left(1-\Omega^{2}\right)^{\frac{1}{2}}\Omega^{-1} \rho_{1}^{\rm{TF}}\left(G^{-\frac{1}{4}}\left(1-\Omega^{2}\right)^{\frac{1}{4}}\Omega^{-\frac{1}{2}}\mathbf{x}\right)
	\end{equation}
	we obtain that
	\begin{equation}\label{energy:TF-scaling}
	E_{\Omega}^{\rm{TF}} = G^{\frac{1}{2}}\left(1-\Omega^{2}\right)^{\frac{1}{2}}\Omega^{-1} \inf \left\{\frac{1}{2}\int_{\mathbb{R}^{2}}e^{\rm Ab}(1)\rho^{2} + |\mathbf{x}|^{2}\rho : \rho \in L^{1}\cap L^{2}(\mathbb{R}^{2} ; \mathbb{R}^{+}), \int_{\mathbb{R}^{2}}\rho=1\right\},
	\end{equation}
	where the minimization problem in~\eqref{energy:TF-scaling} is attained at the (unique) minimizer 
	\begin{equation}\label{eq:TF min}
	 \rho_{1}^{\rm TF} (\bx) = \frac{\left(\lambda^{\rm TF} - |\bx|^2 \right)_+}{e^{\rm Ab}(1)}
	\end{equation}
    where $\lambda^{\rm TF}$ is fixed by normalization. Clearly, these considerations imply
	\begin{equation}\label{supp:TF-min}
	\operatorname{supp}\left(\rho_{\Omega}^{\rm{TF}}\right) \subset B_{CL_{\Omega}^{\rm TF}}(0)  \quad \text{with} \quad L_{\Omega}^{\rm TF} = G^{\frac{1}{4}}\left(1-\Omega^{2}\right)^{-\frac{1}{4}}
	\end{equation}
	for some fixed constant $C>0$, where $B_R(x)$ stands for a ball of radius $R$ centered at $x$. Furthermore, \eqref{energy:behavior-lll} and \eqref{energy:TF-scaling} suggest that, as $\Omega$ tends to $1$, the behavior of the LLL energy \eqref{energy:LLL} is captured correctly at leading order by the Thomas--Fermi type theory. This was conjectured in \cite{AftBlaDal-05,AftBla-06,AftBlaNie-06b}.
	
	The limit $\Omega \nearrow 1$ has mostly been considered at fixed $G$ in the literature. However, the conclusions we aim at must (in view of~ \eqref{energy:behavior-GP-LLL-fake}) stay valid and physically relevant for $G\gg 1$ as long as $G(1-\Omega) \ll 1$. We will choose accordingly an interaction strength $G = G_{\Omega}$ depending on $\Omega$. 
	Our main result is the following:
	
	\begin{theorem}[\textbf{Local density approximation for the rotating gas}]\label{thm:main main}
		~\\
		Let $\#$ denote either $\rm GP$ or $\rm LLL$. Assume $G = G_{\Omega} = \left(1-\Omega^{2}\right)^{-\delta}$ with $-1 < \delta$ if $\#$ denotes $\rm LLL$ and $-1 < \delta < 1$ if $\#$ denotes $\rm GP$.  In the limit $\Omega \nearrow 1$ we have the energy convergence
		\begin{equation}\label{convergence:energy-GP-TF}
		\lim_{\Omega \nearrow 1} \frac{E_{\Omega}^{\#}-1}{E_{\Omega}^{\rm{TF}}} = 1.
		\end{equation}
		Moreover, for any $L^{2}$-normalized function $u^{\#}$ being such that $\mathcal{E}_{\Omega}^{\#}[u^{\#}] = E_{\Omega}^{\#}$, with $\rho^{\#}:=\big|u^{\rm \#}\big|^{2}$, we have for any $R > 0$ and any Lipschitz function $\varphi \in W^{1,\infty} (B_R (0))$
		\begin{equation}\label{convergence:density}
		 \lim_{\Omega \nearrow 1}\int_{\R^2} \left(\LTF\right)^2\rho^{\#}\left(\LTF \bx\right) \varphi(\bx) {\rm d}\bx = \int_{\R^2} \rho_{1}^{\rm TF} \varphi
		\end{equation}
	\end{theorem}
	
	A few comments:
	
	\medskip
	
	\noindent\textbf{1.} The conditions imposed on $G$ in the above are essentially optimal. Indeed, they ensure 
	\begin{equation}\label{eq:phys regime} 
	(1-\Omega) \ll G \ll (1-\Omega)^{-1}.	 
	\end{equation}
	The upper bound $G \ll (1-\Omega)^{-1}$ is needed to ensure that $E_{\Omega}^{\rm TF} \ll 1$, the gap of the Landau operator. Then the $\mathcal{LLL}$ projection is energetically relevant for the full $\rm GP$ problem. This upper bound is not needed for the LLL problem where the projection is imposed. The lower bound $ (1-\Omega) \ll G$ ensures that the length scale $L_{\Omega}^{\rm TF}$ of the Thomas--Fermi problem satisfies $L_{\Omega}^{\rm TF} \gg 1$, i.e. is much larger than the magnetic length (fixed in our units) associated to~\eqref{eq:Landau hamil}. Such a scale decoupling is necessary to rigorously justify a local density approximation (LDA) and the ``homogenization'' of the contribution of the microscopic vortex lattice at the macroscopic scale.
	
	If one drops the condition $ (1-\Omega) \ll G$, the result is certainly wrong. Indeed, it is proved in~\cite{GerGerTho-19} that if $G(1-\Omega)^{-1} < c_0, $ a fixed number (cf~\cite{Schwinte-24} for the sharp threshold), the unique minimizer of the $\mathrm{LLL}$ problem is the gaussian
	$$ 
    \varphi_0 (z) = \frac{1}{\sqrt{\pi}} e^{-|z|^2/2}.
	$$
	The sharpness is proved in~\cite{Schwinte-24} by showing that there exists a constant $c_1 >0$ such that if $c_0 <G(1-\Omega)^{-1} < c_1$ the unique minimizer becomes 
	$$ 
	\varphi_1 (z) = \frac{z}{\sqrt{\pi}} e^{-|z|^2/2}.
	$$
	If $c_1 <G(1-\Omega)^{-1}$, any minimizer has an infinite number of zeros. Our results concern an extreme case of this situation, by giving some (weak) information on the location of theses zeros/vortices (see below).
	
	
	%
	\medskip

	\noindent\textbf{2.} Results related to the above have been obtained in the context of Ginzburg--Landau theory \cite{Almog-06,AftSer-07,FouKac-11,Kachmar-14,SanSer-03}. The main differences are that we have to consider problems with fixed total density, and deal with inhomogeneous systems. Our proof is inspired by the recent study of the local density approximation for the almost-bosonic anyon gas \cite{CorDubLunRou-19,CorLunRou-16}. The thermodynamic energy considered therein however has an exact scaling property, responsible for the occurence of a TF profile in the inhomogeneous problem. We recover the analogue of this scaling law (namely, the fact that the limit $\varrho\to 0$ in~\eqref{eq:def Ab} exists and is non-zero) only in the low density limit, using elliptic estimates.
	
	\medskip

	\noindent\textbf{3.} Energy minimizers for the LLL problem we study provide stationary solutions for the LLL evolution equation 
	$$ \mathrm{i} \partial_t u = \Pi_{0} \left( |u|^{2} u \right),\, u\in \mathcal{LLL}$$
	studied in \cite{GerGerTho-19,GerTho-16,Nier-07}, where $\Pi_{0}$ is the orthogonal projector on $\mathcal{LLL}$. This is because the latter equation conserves interaction energy and angular momentum, and the action of the latter on the lowest Landau level is equivalent to multiplication by $|\mathbf{x}|^{2}$ (see~\cite[Lemma~2.1]{RouSerYng-13b} or~\cite[Lemma~3.1]{AftBla-08}). Our theorem thus provides detailed asymptotic information on some stationary solutions with large angular momentum. Other stationary solutions are investigated and classified in \cite{GerGerTho-19,Schwinte-24}.  
	
	%
	%
	%
	
	%
	%
	\medskip 
	
	The above result can be interpreted as a signature of vortex lattice inhomogeneities as follows. An ansatz for a $\mathcal{LLL}$ wave-function can be written in the form 
	$$ u (\mathbf{x}) = c \prod_{j=1}^J (z-a_{j}) e^{-\frac{|z|^{2}}{2}} $$
	by identifying $\mathbb{R}^{2} \ni \mathbf{x} \leftrightarrow z\in \mathbb{C}$. Here $c$ is a $L^{2}$ normalization constant and $a_{1},\ldots,a_J \in \mathbb{C}$ are the locations of the zeros of the analytic function associated to $u$, cf~\eqref{space:LLL}. It is proved in \cite{AftBlaNie-06b,GerGerTho-19} that, for $\Omega$ sufficiently close to $1$, minimizers are indeed of this form, with $J= \infty$. 
	
	Physically, the points $a_{1},\ldots,a_J \in \mathbb{C}$ correspond to quantized vortices: zeros in the density accompanied by a phase circulation. A remarkable feature of  lowest Landau level wave-functions is a one-to-one (somewhat formal) correspondance \cite{Ho-01,AftBlaDal-05} between the matter density $|u|^{2}$ and the vortex empirical density
	\begin{equation}\label{eq:vort dens}
	\mu := \sum_{j=1} ^J \delta_{a_{j}}.
	\end{equation}
	Namely, using that $(\mathbf{x},\mathbf{y}) \mapsto - (2\pi)^{-1} \log |\mathbf{x}-\mathbf{y}|$ is the Green function of the Laplace operator, we obtain\footnote{The expression is reminiscent of some found in quite different regimes \cite{CorRou-13,SheRad-04,SheRad-04b}.}
	\begin{equation}\label{eq:vort mat}
	\mu = \frac{1}{4\pi} \left( 4 + \Delta (\log |u|^{2}) \right). 
	\end{equation}
	Inserting the TF approximation for the density of the $\mathcal{LLL}$ minimizer derived above leads to conjectural expressions for the latter's vortex density, see the aforementioned references for details. Putting this heuristic on rigorous grounds, even in a weak sense, seems a hard problem, in that asymptotics for the $\log$ of the density should be derived. In any event, the precision of these density asymptotics would probably be sufficient to derive only the leading, constant, bulk contribution to the vortex density, not the edge inhomogeneities numerically observed \cite{AftBlaDal-05,BlaRou-08} and used in \cite{AftBla-06,AftBlaNie-06b} to construct trial states with the correct energy. 
	
	\bigskip
	
	\textbf{Organization of the paper.}
	We start, in Section~\ref{sec:reduc}, by reducing the proof of Theorem~\ref{thm:main main} to the same result under a more stringent, sub-optimal, condition on the parameters. In Section \ref{sec:homogeneous}, we prove the existence of the thermodynamic limit of the homogeneous energy at fixed density. We will show the independence of such a limit from the shape of the domain although we do not need it. The proof of Theorem~\ref{thm:main} is concluded in Section \ref{sec:proof-main-result}. In Appendix \ref{app:boundedness-lll}, we prove the boundedness of the projector onto the finite-dimensional lowest Landau level. This will be used together with elliptic estimates to compute the thermodynamic energy in the low density regime. Appendix \ref{app:gp-lll} contains the estimates between the GP and LLL energies.
	
	\bigskip

	\noindent\textbf{Acknowledgments:} We thank Denis P\'erice for useful discussions and his help with the material of Appendix~\ref{app:boundedness-lll}. Work funded by the European Research Council (ERC) under the European Union's Horizon 2020 Research and Innovation Programme (Grant agreement CORFRONMAT No 758620).

	\section{Reduction to a smaller range of parameters}\label{sec:reduc}
	
	A first, simple, observation is that the $\rm LLL$ minimization problem~\eqref{energy:LLL} essentially only depends on the ratio $G(1-\Omega)^{-1}$. Under~\eqref{eq:phys regime}, so should the $\rm GP$ problem (at the level of precision we aim at). At least, the truth of Theorem~\ref{thm:main main} in the $\rm LLL$ case must depend only on asymptotic bounds on $G(1-\Omega)^{-1}$. A direct approach to the LDA by Dirichlet-Neumann bracketing however leads to more stringent conditions:

	\begin{theorem}[\textbf{Local density approximation with restricted parameters}]\label{thm:main}
		~\\
		Let $\#$ denote either $\rm GP$ or $\rm LLL$. Assume $G = G_{\Omega} = \left(1-\Omega^{2}\right)^{-\delta}$ with $\frac{3}{5} < \delta < 1$. Then the conclusions of Theorem~\ref{thm:main main} hold. Moreover~\eqref{convergence:density} holds for $\rho^{\#}:=\big|u^{\#}\big|^{2}$ under the weaker assumption that 
		\begin{equation}\label{eq:quasimin}
		\mathcal{E}_{\Omega}^{\#}[u^{\#}] = E_{\Omega}^{\#} + o\left( \ETF \right)
		\end{equation}
		in the limit $\Omega \nearrow 1.$
	\end{theorem}
	
	The core of our analysis is the proof of Theorem~\ref{thm:main}. Before proceeding to this we explain in this section how to deduce the general case of Theorem~\ref{thm:main main} from Theorem~\ref{thm:main} and the considerations of Appendix~\ref{app:gp-lll}. 
	Briefly, the scheme is as follows:
	
	\medskip
	\noindent \textbf{1.} First, we use the fact that the LLL problem depends only on $G(1-\Omega)^{-1}$. Upon a simple change of parameters respecting this invariance, we find that the general LLL problem with $-1<\delta <1$ is equivalent to an instance of Theorem~\ref{thm:main} where $\frac{3}{5} < \delta < 1$. This gives the proof of Theorem~\ref{thm:main main} in the general LLL case.

	\medskip
	\noindent \textbf{2.} We next prove that, provided $G\ll (1-\Omega)^{-1}$, the GP problem reduces to the LLL one at the level of precision we aim at. This follows ideas from~\cite{AftBla-08}, that we supplement in Appendix~\ref{app:gp-lll} with a new elliptic estimate in order to cover the full physical range of parameters.  
	
	\medskip
	
	Hence we reduce the GP problem with the physical parameters to  the LLL problem, and then use the invariance of the latter to change the value of the parameters to more tractable ones. Namely, a different GP problem, with tamed parameters is used to study the physical LLL problem. 
	
	It will be apparent from the proof of Theorem~\ref{thm:main} that such a reduction is necessary within our approach (see in particular Remarks~\ref{rem:NeuDir} and~\ref{rem:LDA} below). Relaxing a LLL problem to a GP one is convenient in that it allows to bypass the extreme rigidity of $\mathcal{LLL}$ trial functions. The point of this section is that the relaxation needs not be performed with the challenging physical values of the parameters. 
	
	\begin{proof}[Proof of Theorem~\ref{thm:main main} given Theorem~\ref{thm:main}]
	 In this proof we denote 
	 $$ h = 1-\Omega^{2}$$
	 for short and recall that we set $G=h^{-\delta}$. Let $\alpha\in \R$ to be fixed later on. Writing the LLL functional in the manner
	 \begin{align*}
	\mathcal{E}_{\Omega}^{\rm LLL}[u] &= 1 + \frac{1}{2}\int_{\mathbb{R}^{2}} G|u|^{4} + \frac{1-\Omega^{2}}{\Omega^{2}}|\mathbf{x}|^{2}|u|^{2}\\
	&= 1 + h^{-\alpha} \Omega^{-2} \left(\frac{1}{2}\int_{\mathbb{R}^{2}} \Omega^{2} h^{\alpha-\delta}|u|^{4} + h ^{1+\alpha}|\mathbf{x}|^{2}|u|^{2}\right).
	\end{align*}
	It is apparent that it has the same minimizers as 
	 $$ 
	 \mathcal{LLL} \in u \mapsto 1 + \left(\frac{1}{2}\int_{\mathbb{R}^{2}} \Omega^{2} h^{\alpha-\delta}|u|^{4} + h ^{1+\alpha}|\mathbf{x}|^{2}|u|^{2}\right)
	 $$
	 and that the asymptotics of $E_{\Omega}^{\rm LLL}-1$ can be deduced from studying those of 
	 $$ 
	 \inf \left\{\frac{1}{2}\int_{\mathbb{R}^{2}} \Omega^{2} h^{\alpha-\delta}|u|^{4} + h ^{1+\alpha}|\mathbf{x}|^{2}|u|^{2}, u\in \mathcal{LLL}, \int_{\R^2} |u|^2 =1 \right\}
	 $$
	 and then multiplying by $h^{-\alpha}\Omega^{-2}$. Writing 
	 $$ h^{\alpha - \delta} = h ^{(1+\alpha)\frac{\alpha - \delta}{1+\alpha}}$$
	 Theorem~\ref{thm:main} applies to this problem\footnote{The factor $\Omega^{-2}$ multiplying the second term of~\eqref{functional:LLL} plays no role in the limit $\Omega \nearrow 1$ that we consider.} mutatis mutandis provided 
	 $$ 1+\alpha > 0$$
	 so that $h^{1+\alpha} \to 0$ as $h\to 0$ and 
	 $$ \frac{3}{5} < \frac{\delta-\alpha}{1+\alpha} < 1.$$
	 Then the scaling assumption regarding the parameters in Theorem~\ref{thm:main} is satisfied. Given $\delta$ we look for $\alpha$ allowing to meet the above requirements. The constraints reduce to 
	 $$ \alpha > - 1, \quad \alpha > \frac{\delta -1}{2} \quad \text{and} \quad \alpha < \frac{5\delta - 3}{8}$$
	 so that, if $\delta > -1$ one can always find some $\alpha \in \R$ satisfying all three constraints, e.g.,
	 $$
	 \alpha = \frac{1}{2}\left(\frac{\delta -1}{2} + \frac{5\delta - 3}{8}\right) = \frac{9\delta - 7}{16}.
	 $$
	 We hence deduce the part of Theorem~\ref{thm:main main} bearing on the LLL problem from Theorem~\ref{thm:main} by making such a choice of $\alpha$. Moreover~\eqref{convergence:density} holds in this case under the weaker condition~\eqref{eq:quasimin}.

	 The GP case follows for $-1 < \delta < 1$ by combining with Proposition~\ref{pro:LLL-infinite}. The energy estimate is immediate from~\eqref{energy:behavior-GP-LLL}. Let $\Pi_0$ be the orthogonal projector on $\LLL$. The equivalent of~\eqref{convergence:density} with $\uGP$ replaced by 
	 $$\frac{\Pi_0 \uGP}{\norm{\Pi_0 \uGP}_{L^2}}$$ 
	 is obtained by combining~\eqref{eq:dens GP to LLL 2} with the observation that~\eqref{convergence:density} holds for $\#=\mathrm{LLL}$ under~\eqref{eq:quasimin}. There only remains to use~\eqref{eq:dens GP to LLL 1} to leverage the result to $\uGP$ itself. 
	\end{proof}

	\section{The homogeneous gas in the thermodynamic limit}\label{sec:homogeneous}
	
	%
	%
	%
	
	We start by putting the definition~\eqref{eq:def Ab} on rigorous ground. The existence of the thermodynamic limit $L\to \infty$ is proved in Subsection~\ref{sec:L to inf} and the low density regime $\varrho\to 0$ is considered in Subsection~\ref{sec:rho to 0}. In both cases we need precise quantitative estimates as input in our analysis of the inhomogeneous problem.
	
	\subsection{Existence of the thermodynamic limit}\label{sec:L to inf}
	
	We first discuss the large-volume limit for the homogeneous gas. Let $\mathcal{D}$ be a fixed bounded domain in $\mathbb{R}^{2}$, with the associated Neumann energy $E^{\rm GP}(\mathcal{D},M)$ given by \eqref{energy:homogeneous-neumann}.
	We also define the following energy with homogeneous Dirichlet boundary condition
	\begin{equation}\label{energy:homogeneous-dirichlet}
	E_{0}^{\rm GP}(\mathcal{D},M) = \inf\left\{\mathcal{E}_{\mathcal{D}}^{\rm GP}[u] : u\in H_{0}^{1}(\mathcal{D}), \int_{\mathcal{D}}|u|^{2} = M\right\},
	\end{equation}
	where $\mathcal{E}_{\mathcal{D}}^{\rm GP}[u]$ is defined by \eqref{functional:homogeneous-neumann}.
	
	In this subsection, we show that the thermodynamic limit exists and does not depend on boundary conditions. This is a crucial ingredient in our study of the trapped case. 
	
	\begin{theorem}[\textbf{Thermodynamic limit for the homogeneous gas}]\label{thm:thermodynamic-limit-general}
		~\\
		Let $\mathcal{D} \subset \mathbb{R}^{2}$ be a bounded simply connected domain with Lipschitz boundary, $G > 0$ and $\varrho > 0$ be fixed parameters. Then, the limits
		\begin{equation}\label{limit:thermodynamic}
		e^{\rm GP}(\varrho) := \lim_{L \to \infty} \frac{E^{\rm GP}(L\mathcal{D},\varrho |L\mathcal{D}|)}{|L\mathcal{D}|} = \lim_{L \to \infty} \frac{E_{0}^{\rm GP}(L\mathcal{D},\varrho |L\mathcal{D}|)}{|L\mathcal{D}|}
		\end{equation}
		exist and coincide.
	\end{theorem}
	
	%
	%
	
	We prove the existence of the thermodynamic limit for the case of squares. Although this is enough for the proof of our main results, it is of interest to extend the result to general domains. We need the following lemmas.

	\begin{lemma}[\textbf{Uniform bounds on the GP energy per area}]\label{lem:thermodynamic-limit-boundedness}
		~\\
		For any fixed bounded domain $\mathcal{D}$ and $G, \varrho > 0$, there exists a constant $C > 0$ such that
		$$
		\frac{E^{\rm GP}(L\mathcal{D},\varrho |L\mathcal{D}|)}{|L\mathcal{D}|} \leq \frac{E_{0}^{\rm GP}(L\mathcal{D},\varrho |L\mathcal{D}|)}{|L\mathcal{D}|} \leq C
		$$
		for all $L \geq 1$.
	\end{lemma}
	
	\begin{proof}
		Since $H_{0}^{1} \subseteq H^{1}$, we obviously have the first inequality. Let us prove the second one. We fill the domain $L\mathcal{D}$ with $N \sim L^{2}$ disks on which we use fixed trial states with Dirichlet boundary conditions. Let $f \in C_c^\infty(B_{1}(0);\mathbb{R}^+)$ be a smooth, compactly supported, radial function with $\int_{B_{1}(0)}|f|^{2} = 1$, and let
		$$
		u_{j}(\mathbf{x}) := \sqrt{\omega_{N}}f(\mathbf{x}-\mathbf{x}_{j}) \in C_c^\infty(B_{1}(\mathbf{x}_{j})) \quad \text{with} \quad \omega_{N} = \frac{\varrho |L\mathcal{D}|}{N}.
		$$
		Here the points $\mathbf{x}_{j}$, $j =1,\ldots,N$, are distributed in $L\mathcal{D}$ in such a way that the disks $B_{1}(\mathbf{x}_{j})$ are contained in $L\mathcal{D}$ and disjoint, with $N \sim c|L\mathcal{D}|$ as $L \to \infty$ for some $c > 0$. Hence
		$$
		\lim_{N\to\infty}\omega_{N} = \frac{\varrho}{c}.
		$$
		We note that
		$$
		{\rm curl}(\mathbf{x}^{\perp} - (\mathbf{x}-\mathbf{x}_{j})^\perp) = 0.
		$$
		Thus there exists a gauge phase $\phi_{j} = \mathbf{x}_{j}^{\perp} \cdot \mathbf{x}$ on $B_{1}(\mathbf{x}_{j})$ such that
		$$
		\mathbf{x}^{\perp} - (\mathbf{x}-\mathbf{x}_{j})^{\perp} = \nabla \phi_{j} \quad \text{in} \quad  B_{1}(\mathbf{x}_{j}).
		$$
		Take then the trial state
		$$
		u := \sum_{j=1}^{N} e^{\mathrm{i}\phi_{j}}u_{j} \in C_c^\infty(L\mathcal{D}).
		$$
		Note that
		$$
		\int_{\mathbb{R}^{2}}|u|^{2} = \sum_{j=1}^{N}\int_{B_{1}(\mathbf{x}_{j})}|u_{j}|^{2} = N\omega_{N}\int_{B_{1}(0)}|f|^{2} = \varrho |L\mathcal{D}|.
		$$
		Then
		\begin{align}
		E_{0}^{\rm GP}(L\mathcal{D},\varrho |L\mathcal{D}|) \leq \mathcal{E}_{L\mathcal{D}}^{\rm GP}[u] & = \sum_{j=1}^{N} \frac{1}{2}\int_{B_{1}(\mathbf{x}_{j})} \big|\big(-\mathrm{i}\nabla - \mathbf{x}^{\perp}\big)e^{\mathrm{i}\phi_{j}}u_{j}\big|^{2} + G|e^{\mathrm{i}\phi_{j}}u_{j}|^{4} \nonumber \\
		& = \sum_{j=1}^{N} \frac{1}{2}\int_{B_{1}(\mathbf{x}_{j})} \big|\big(-\mathrm{i}\nabla - \mathbf{x}^{\perp} + \nabla \phi_{j}\big)u_{j}\big|^{2} + G|u_{j}|^{4} \nonumber \\
		& = \sum_{j=1}^{N} \frac{1}{2}\int_{B_{1}(\mathbf{x}_{j})} \big|\big(-\mathrm{i}\nabla - \big(\mathbf{x}-\mathbf{x}_{j}\big)^\perp\big)u_{j}\big|^{2} + G|u_{j}|^{4} \nonumber \\
		& = \frac{N\omega_{N}}{2}\int_{B_{1}(0)} \big|\big(-\mathrm{i}\nabla - \mathbf{x}^{\perp}\big)f\big|^{2} + G\omega_{N}|f|^{4} \nonumber \\
		& \leq C\varrho(1+G\varrho)|L\mathcal{D}|\label{upper:estimate-Neu}
		\end{align}
		for some large enough constant $C > 0$ independent of $L$.
	\end{proof}
	
	\begin{remark}[Bounds on GP energy]~\\
		Although the bound in Lemma \ref{lem:thermodynamic-limit-boundedness} is enough for our proof of Theorem \ref{thm:thermodynamic-limit-general}, we need to better bound the GP energy in order to perform the LDA in Section \ref{sec:proof-main-result}. In Theorem \ref{thm:Abrikosov} below, by estimating $E^{\rm GP}(L\mathcal{D},\varrho |L\mathcal{D}|)$ via the GP energy with ``periodic'' boundary condition, we obtain
		\begin{equation}\label{upper:estimate-Dir}
		E^{\rm GP}(L\mathcal{D},\varrho |L\mathcal{D}|) \leq \varrho(1+CG\varrho)|L\mathcal{D}|,
		\end{equation}
		for some constant $C > 0$ independent of $L$.\hfill$\diamond$
	\end{remark}
	
	In order to show that the thermodynamic limit does not depend on boundary conditions, we need to perform energy localizations using an IMS type inequality.
	
	\begin{lemma}[\textbf{IMS formula}]\label{lem:ims}
		~\\	
		Let $\mathcal{D}\subseteq \mathbb{R}^{2}$ be a domain with Lipschitz boundary and $\chi^{2}+\eta^{2}=1$ be a partition of unity such that $\chi$ and $\eta$ are real valued, $\chi \in C_{c}^{\infty}(\mathcal{D})$ and $\operatorname{supp} \chi$ is simply connected. Then, for any $u \in H^{1}(\mathcal{D})$, we have
		\begin{equation}\label{eq:ims}
		\mathcal{E}_{\mathcal{D}}^{\rm GP}[u] = \mathcal{E}_{\mathcal{D}}^{\rm GP}[\chi u] + \mathcal{E}_{\mathcal{D}}^{\rm GP}[\eta u] + G\int_{\mathcal{D}}\chi^2\eta^2|u|^{4}- \int_{\mathcal{D}}\left(|\nabla \chi|^{2}+|\nabla \eta|^{2}\right)|u|^{2}.
		\end{equation}
	\end{lemma}
	
	\begin{proof} We expand
		$$
		\mathcal{E}_{\mathcal{D}}^{\rm GP}[u] = \frac{1}{2}\int_{\mathcal{D}}|\nabla u|^{2} + \mathbf{x}^{\perp} \cdot \mathbf{J}[u] + |\mathbf{x}|^{2}|u|^{2} + G|u|^{4}
		$$
		where
		$$
		\mathbf{J}[u] = \mathrm{i}(u\overline{\nabla u}-\overline{u}\nabla u).
		$$
		For the first term we use the standard IMS formula \cite[Theorem 3.2]{CycFroKirSim-87}, while for the term involving $\mathbf{J}$ we have, using that $\chi$ and $\eta$ are real valued,
		$$
		\begin{aligned}\label{ineq:ims-2}
		\frac{1}{\mathrm{i}}(\mathbf{J}[\chi u] + \mathbf{J}[\eta u]) &=u \chi \nabla(\chi \bar{u})+u \eta \nabla(\eta \bar{u})-\bar{u} \chi \nabla(\chi u)-\bar{u} \eta \nabla(\eta u) \\
		&=u\left(\chi^{2}+\eta^{2}\right) \nabla \bar{u}-\bar{u}\left(\chi^{2}+\eta^{2}\right) \nabla u=\frac{1}{\mathrm{i}} \mathbf{J}[u].
		\end{aligned}
		$$
		Finally, for the last term we use the identity
		$$
		1 = (\chi^{2}+\eta^{2})^{2}.
		$$
		We can then recollect the terms to obtain \eqref{eq:ims}.
	\end{proof}
	
	\begin{lemma}[\textbf{Dirichlet--Neumann comparison}]\label{lem:NeuDir}
		~\\
		Let $\mathcal{D}$ be a bounded simply connected domain with Lipschitz boundary. Then, for any fixed positive parameters $G$ and $\varrho$, there exists a constant $C > 0$ such that
		\begin{equation}\label{upper:estimate-DirNeu}
		\begin{aligned}
		E_{0}^{\rm GP}(L\mathcal{D},\varrho |L\mathcal{D}|) & \geq E^{\rm GP}(L\mathcal{D},\varrho |L\mathcal{D}|) \\
		& \geq E_{0}^{\rm GP}(L\mathcal{D},\varrho |L\mathcal{D}|) - C(1+G\varrho)\left(LG^{-1} + \varrho L^{\frac{3}{2}}\right).
		\end{aligned}
		\end{equation}
	\end{lemma}
	
	\begin{proof}
		
		The first inequality in the statement is trivial. It remains to prove the second inequality. We need to make an IMS localization on a small enough region, and therefore consider a division of $L\mathcal{D}$ into a bulk region surrounded by a thin shell close to the boundary. For this purpose, we will use the length scale 
		$$
		\ell \ll L.
		$$
		Let $Q_{\ell}$ be a shell of width $\ell > 0$ closest to the boundary of $L\mathcal{D}$, i.e.,
		$$
		Q_{\ell}:=\left\{x \in L\mathcal{D}: \operatorname{dist}(x, \partial(L\mathcal{D})) < \ell\right\} .
		$$
		Let $u \in H^{1}(L\mathcal{D})$ be a minimizer for $E^{\rm GP}(L\mathcal{D},\varrho |L\mathcal{D}|)$. We now perform an IMS localization on $Q_{\ell}$. We pick a partition $\chi^{2}+\eta^{2}=1$, such that $\chi$ varies smoothly from $1$ to $0$ outwards on $Q_{\ell}$, so that $\chi=1$ (resp. $\eta=1$) on the inner (resp. outer) component of $Q_{\ell}^{c}$. By Lemma~\ref{lem:ims}, we have
		\begin{align}\label{energy:deviation}
		\mathcal{E}_{L\mathcal{D}}^{\rm GP}[u] & = \mathcal{E}_{L\mathcal{D}}^{\rm GP}[\chi u] + \mathcal{E}_{L\mathcal{D}}^{\rm GP}[\eta u] + G\int_{L\mathcal{D}}\chi^{2}\eta^{2}|u|^{4} - \int_{Q_{\ell}}\left(|\nabla \chi|^{2}+|\nabla \eta|^{2}\right)|u|^{2} \\
		& \geq \int_{L\mathcal{D}}\chi^{2}|u|^{2} + \frac{G}{2}\int_{L\mathcal{D}}|u|^{4} - C\ell^{-2}\int_{Q_{\ell}}|u|^{2}. \nonumber
		\end{align}
		Choosing $\ell \sim 1$ and using \eqref{upper:estimate-Dir}, we obtain
		$$
		CG\varrho^{2}L^{2} + C\int_{Q_{\ell}}|u|^{2} \geq \frac{G}{2}\int_{Q_{\ell}}|u|^{4} \geq CGL^{-1}\left(\int_{Q_{\ell}}|u|^{2}\right)^{2}.
		$$
		This implies that we must have
		\begin{equation}\label{mass:deviation}
		\int_{Q_{\ell}}|u|^{2} \leq C\left(LG^{-1} + \varrho L^{\frac{3}{2}}\right).
		\end{equation}
		The above implies that the mass of $\chi^{2}|u|^{2}$ is very close to $\varrho |L\mathcal{D}| = \int_{L\mathcal{D}}|u|^{2}$. 
		
		On the other hand, we denote
		$$
		v = \left(\frac{\varrho|L\mathcal{D}|}{\int_{L\mathcal{D}}\chi^{2}|u|^{2}}\right)^{\frac{1}{2}} \chi u.
		$$
		Then $v \in H^{1}_{0}(L\mathcal{D})$ with $\int_{L\mathcal{D}}|v|^{2} = \varrho|L\mathcal{D}|$ and we have
		\begin{align}\label{energy:rescale}
		\mathcal{E}_{L\mathcal{D}}^{\rm GP}[\chi u] & = \frac{1}{2}\int_{L\mathcal{D}}\frac{\int_{L\mathcal{D}}\chi^{2}|u|^{2}}{\varrho |L\mathcal{D}|} \big|\big(-\mathrm{i}\nabla - \mathbf{x}^{\perp}\big)v\big|^{2} + \left(\frac{\int_{L\mathcal{D}}\chi^{2}|u|^{2}}{\varrho |L\mathcal{D}|}\right)^{2}G|v|^{4} \nonumber \\
		& \geq \min\left\{\frac{\int_{L\mathcal{D}}\chi^{2}|u|^{2}}{\varrho |L\mathcal{D}|},\left(\frac{\int_{L\mathcal{D}}\chi^{2}|u|^{2}}{\varrho |L\mathcal{D}|}\right)^{2}\right\}\mathcal{E}_{L\mathcal{D}}^{\rm GP}[v] \nonumber \\
		& \geq \left(\frac{\int_{L\mathcal{D}}\chi^{2}|u|^{2}}{\varrho |L\mathcal{D}|}\right)^{2} E_{0}^{\rm GP}(L\mathcal{D},\varrho |L\mathcal{D}|).
		\end{align}
		Now we use \eqref{upper:estimate-Neu}, \eqref{energy:deviation}, \eqref{mass:deviation} and \eqref{energy:rescale} to obtain that
		\begin{align*}
		E^{\rm GP}(L\mathcal{D},\varrho |L\mathcal{D}|) = \mathcal{E}_{L\mathcal{D}}^{\rm GP}[u] & \geq \left(1 - 2\frac{\int_{L\mathcal{D}}\eta^{2}|u|^{2}}{\varrho|L\mathcal{D}|}\right) E_{0}^{\rm GP}(L\mathcal{D},\varrho |L\mathcal{D}|) - C\int_{Q_{\ell}}|u|^{2} \\
		& \geq E_{0}^{\rm GP}(L\mathcal{D},\varrho |L\mathcal{D}|) - C(1+G\varrho)\left(LG^{-1} + \varrho L^{\frac{3}{2}}\right).
		\end{align*}
		This completes the proof.
	\end{proof}
	
	\begin{remark}[Dirichlet--Neumann comparison]\label{rem:NeuDir}~\\
		There is probably room for improvement in our bounds, but we certainly expect that the Dirichlet and Neumann energy must differ by at least a $O(L G^{-1})$ for low densities (a regime we will focus on in the next subsection). Here is why. 
		
		Consider the magnetic Laplacian 
		$$ \frac{1}{2} \left(-\mathrm{i}\nabla + \mathbf{A} \right)^{2}$$
		for constant magnetic field $B = - \mathrm{curl} \, \mathbf{A} = 2$. Low kinetic energies are obtained by localizing trial states on the order of the magnetic length, fixed in these units. Localization away from the boundary leads to an energy $\sim M$ at mass $M$, as one would obtain from the full space Landau Hamiltonian~\eqref{eq:Landau hamil}. Localization close to the boundary however leads to an energy $\sim \Theta_{0} M < M$ with $\Theta_{0}$ being the de Gennes constant, connected to the realization of~\eqref{eq:Landau hamil} on a half-plane with Neumann conditions on the boundary. We refer to \cite{FouHel-book} for background and theorems on these well-known facts. Note that they immediately impose conditions on $G$ for Dirichlet and Neumann energies to coincide: the theorem is certainly wrong for $G=0$.
		
		To get a heuristic estimate on the difference between Dirichlet and Neumann energies, start from a fully homogeneous system with density $\varrho$ and consider increasing the density in a shell of thickness $\sim 1$ close to the boundary by moving some mass $M$ from the bulk. If 
		$$ \varrho L \ll M \ll \varrho L^{2}$$
		we barely change the bulk density, but increase a lot the boundary density, at a cost of roughly  
		$$ G M^{2} L^{-1}$$
		in interaction energy. If we use Dirichlet conditions such a move is forbidden. But if we use Neumann boundary conditions, it is not only authorized but it can bring a gain of  
		$$ M (1-\Theta_{0}) \propto M$$
		in magnetic kinetic energy, as per the above discussion. Choosing $M\sim L G^{-1}$ to balance gain and loss we expect that the Neumann energy must include a negative term of order $\sim  L G ^{-1}$, absent from the Dirichlet energy. This is due to a larger boundary density in the Neumann case, favored by the spectral properties of the Landau Hamiltonian recalled above. 
		
		The error $O(L G ^{-1})$ is the main motivation to use the approach of Section~\ref{sec:reduc} to obtain Theorem~\ref{thm:main main} in full generality. It leads to the constraint $G \gg (1-\Omega)^{-\frac{3}{5}}$ in Theorem~\ref{thm:main} when performing the LDA, as further discussed in Remark~\ref{rem:LDA} below.\hfill$\diamond$
	\end{remark}

	\begin{lemma}[\textbf{Thermodynamic limit for the Dirichlet energy in a square}]\label{lem:thermodynamic-limit-square}~\\
		Let $K_{L}$ be a square of side length $L>0$, centered at the origin, $G > 0$ and $\varrho > 0$ be fixed parameters. The limit
		$$
		e^{\rm GP}(\varrho) = \lim_{L\to\infty}\frac{E_{0}^{\rm GP}\left(K_{L},\varrho L^{2}\right)}{L^{2}}
		$$
		exists and is finite.	
	\end{lemma}
	
	\begin{proof}
		Let $\left(L_{n}\right)_{n \in \mathbbm{N}}$ and $\left(L_{m}\right)_{m \in \mathbbm{N}}$ be two increasing sequences of positive real numbers such that $L_{n} \to \infty, L_{m} \to \infty$ and
		\begin{align*}
		\lim_{n\to\infty}\frac{E_{0}^{\rm GP}\left(K_{L_{n}},\varrho L_{n}^{2}\right)}{L_{n}^{2}} & = \liminf_{L \to \infty} \frac{E_{0}^{\rm GP}\left(K_{L},\varrho L^{2}\right)}{L^{2}},\\	\lim_{m\to\infty}\frac{E_{0}^{\rm GP}\left(K_{L_{m}},\varrho L_{m}^{2}\right)}{L_{m}^{2}} & = \limsup _{L \to \infty} \frac{E_{0}^{\rm GP}\left(K_{L},\varrho L^{2}\right)}{L^{2}}.
		\end{align*}
		For each $n$, there must exist a sequence of integers
		$$
		q_{nm} \to+\infty \quad \text {as} \quad m \to \infty
		$$
		such that, for $m$ large enough, e.g., $m \gg n$,
		$$
		L_{m}=q_{nm} L_{n}+k_{n m}, \quad 0 \leq k_{n m}<L_{n}.
		$$
		We build a trial state for $E_{0}^{\rm GP}\left(K_{L_{m}},\varrho L_{m}^{2}\right)$ as follows. The square $K_{L_{m}}$ must contain $q_{nm}^{2}$ disjoint squares of side length $L_{n}$ that we denote by $K_{L_{nj}}, j=1, \ldots, q_{nm}^{2}$. On the remaining part of the domain we can construct, as in the proof of Lemma~\ref{lem:thermodynamic-limit-boundedness}, a function $\tilde{u}_{0}$ of mass $\varrho (L_{m}^{2}-q_{nm}^{2}L_{n}^{2})$ with compact support in $K_{L_{m}} \backslash \bigcup_{j=1}^{q_{nm}^{2}} K_{L_{nj}}$, satisfying
		$$
		\mathcal{E}_{K_{L_{m}}}^{\rm GP}[\tilde{u}_{0}] \leq C\left(L_{m}^{2}-q_{nm}^{2} L_{n}^{2}\right) \leq C L_{m} k_{n m}.
		$$
		We define an admissible trial state
		$$
		u:=\sum_{j=1}^{q_{nm}^{2}} e^{\mathrm{i}\phi_{j}}u_{j} + \tilde{u}_{0}
		$$
		where
		$$
		u_{j}(\mathbf{x}) = u_{0}(\mathbf{x}-\mathbf{x}_{j})
		$$
		with $u_{0}$ a minimizer for $E_{0}^{\rm GP}\left(K_{L_{n}},\varrho L_{n}^{2}\right)$, and $\mathbf{x}_{j}$ the center points of $K_{L_{nj}}$. The phases $\phi_{j}$ are chosen in such a way that
		$$
		\mathbf{x}^{\perp} - (\mathbf{x}-\mathbf{x}_{j})^{\perp}= \nabla \phi_{j} \quad \text{in} \quad  K_{L_{nj}}.
		$$
		Computing the energy, we have
		\begin{align*}
		\mathcal{E}_{K_{L_{m}}}^{\rm GP}[u] = \sum_{j=1}^{q_{nm}^{2}} \mathcal{E}_{K_{L_{nj}}}^{\rm GP}[e^{\mathrm{i}\phi_{j}}u_{j}] + \mathcal{E}_{K_{L_{m}}}^{\rm GP}[\tilde{u}_{0}] &= \sum_{j=1}^{q_{nm}^{2}} \mathcal{E}_{K_{L_{n}}}^{\rm GP}[u_{0}] + \mathcal{E}_{K_{L_{m}}}^{\rm GP}[\tilde{u}_{0}] \\
		& = q_{nm}^{2} E_{0}^{\rm GP}\left(K_{L_{n}},\varrho L_{n}^{2}\right)+\mathcal{O}\left(L_{m} k_{n m}\right).
		\end{align*}
		Since $\int_{K_{L_{m}}}|u|^{2} = \varrho L_{m}^{2}$, it follows from the variational principle that
		$$
		\frac{E_{0}^{\rm GP}\left(K_{L_{m}},\varrho L_{m}^{2}\right)}{L_{m}^{2}} \leq \frac{E_{0}^{\rm GP}\left(K_{L_{n}},\varrho L_{n}^{2}\right)}{L_{n}^{2}}\left(1+\mathcal{O}\left(\frac{k_{n m}}{L_{m}}\right)\right)+\mathcal{O}\left(\frac{k_{n m}}{L_{m}}\right)
		$$
		where we have used the fact that
		$$
		q_{nm}^{2}=\frac{L_{m}^{2}}{L_{n}^{2}}\left(1-\frac{k_{n m}}{L_{m}}\right)^{2}.
		$$
		Passing to the limit $m \to \infty$ first and then $n \to \infty$ yields
		$$
		\limsup _{L \to \infty} \frac{E_{0}^{\rm GP}\left(K_{L},\varrho L^{2}\right)}{L^{2}} \leq\liminf _{L \to \infty} \frac{E_{0}^{\rm GP}\left(K_{L},\varrho L^{2}\right)}{L^{2}}
		$$
		and thus the limit exists.
	\end{proof}

	Now we are in the position to construct the thermodynamic limit in the general case.
	
	\begin{proof}[Proof of Theorem~\ref{thm:thermodynamic-limit-general}] The result is proven as usual by comparing suitable upper and lower bounds to the energy.
		
		\textbf{Upper bound.} We cover $L\mathcal{D}$ with squares $K_{j}, j=1, \ldots, N_{\ell}$, of side length $\ell=L^{\eta}, 0<\eta<1$, retaining only the squares completely contained in $L\mathcal{D} .$ One can estimate the area not covered by such squares as
		\begin{equation}\label{ineq-upper:thermodynamic-limit-general-1}
		\left|L\mathcal{D} \backslash\left(\bigcup_{j=1}^{N_{\ell}} K_{j}\right)\right| \leq C \ell L=o(L^{2}).
		\end{equation}
		Then we define the trial state
		$$
		u := \sum_{j=1}^{N_{\ell}} e^{\mathrm{i}\phi_{j}}u_{j}
		$$
		where
		$$
		u_{j}(\mathbf{x}):=u_{0}\left(\mathbf{x}-\mathbf{x}_{j}\right)\mathbbm{1}_{K_{j}},
		$$
		with $u_{0}$ a minimizer for the Dirichlet problem with mass $\varrho |L\mathcal{D}|N_{\ell}^{-1}$ in a square $K_{\ell}$ of side length $\ell$, centered at the origin, and $\mathbf{x}_{j}$ the center point of $K_{j}$. The phases $\phi_{j}$ are chosen in such a way that
		$$
		\mathbf{x}^{\perp} - (\mathbf{x}-\mathbf{x}_{j})^{\perp}= \nabla \phi_{j} \quad \text{in} \quad  K_{j}.
		$$
		Note that 
		$$
		\int_{L\mathcal{D}}|u|^{2} = \sum_{j=1}^{N_{\ell}}\int_{K_{j}}|u_{j}|^{2} =  N_{\ell}\int_{K_{\ell}}|u_{0}|^{2} = \varrho |L\mathcal{D}|.
		$$
		Hence, it follows from the variational principle that
		$$
		E_{0}^{\rm GP}(L\mathcal{D},\varrho |L\mathcal{D}|) \leq \mathcal{E}_{L\mathcal{D}}^{\rm GP}[u] = \sum_{j=1}^{N_{\ell}} \mathcal{E}_{K_{j}}^{\rm GP}[e^{\mathrm{i}\phi_{j}}u_{j}] = \sum_{j=1}^{N_{\ell}} \mathcal{E}_{K_{\ell}}^{\rm GP}[u_{0}] = N_{\ell}E_{0}^{\rm GP}\left(K_{\ell},\varrho |L\mathcal{D}| N_{\ell}^{-1}\right).
		$$
		By changing variables
		$$
		v = \left(\frac{\ell^{2}N_{\ell}}{|L\mathcal{D}|}\right)^{\frac{1}{2}}u,
		$$
		we obtain
		\begin{align*}
		& E_{0}^{\rm GP}\left(K_{\ell},\varrho |L\mathcal{D}| N_{\ell}^{-1}\right) \\
		= & \inf\left\{\frac{1}{2}\int_{K_{\ell}} \big|\big(-\mathrm{i}\nabla - \mathbf{x}^{\perp}\big)u\big|^{2} + G|u|^{4} : u\in H^{1}_{0}(K_{\ell}), \int_{K_{\ell}}|u|^{2} = \varrho |L\mathcal{D}| N_{\ell}^{-1}\right\} \\
		= & \inf\left\{\frac{1}{2}\int_{K_{\ell}}\frac{|L\mathcal{D}|}{\ell^{2}N_{\ell}} \big|\big(-\mathrm{i}\nabla - \mathbf{x}^{\perp}\big)v\big|^{2} + \left(\frac{|L\mathcal{D}|}{\ell^{2}N_{\ell}}\right)^{2}G|v|^{4} : v\in H^{1}_{0}(K_{\ell}), \int_{K_{\ell}}|v|^{2} = \varrho\ell^{2}\right\} \\
		\leq & \max\left\{\frac{|L\mathcal{D}|}{\ell^{2}N_{\ell}},\left(\frac{|L\mathcal{D}|}{\ell^{2}N_{\ell}}\right)^{2}\right\} \inf\left\{\frac{1}{2}\int_{K_{\ell}} \big|\big(-\mathrm{i}\nabla - \mathbf{x}^{\perp}\big)v\big|^{2} + G|v|^{4} : v\in H^{1}_{0}(K_{\ell}), \int_{K_{\ell}}|v|^{2} = \varrho\ell^{2}\right\} \\
		= & \max\left\{\frac{|L\mathcal{D}|}{\ell^{2}N_{\ell}},\left(\frac{|L\mathcal{D}|}{\ell^{2}N_{\ell}}\right)^{2}\right\} E_{0}^{\rm GP}\left(K_{\ell},\varrho\ell^{2}\right).
		\end{align*}
		Thus, we conclude that
		\begin{equation}\label{ineq-upper:thermodynamic-limit-general-2}
		\frac{E_{0}^{\rm GP}(L\mathcal{D},\varrho |L\mathcal{D}|)}{|L\mathcal{D}|} \leq \frac{\ell^{2}N_{\ell}}{|L\mathcal{D}|} \max\left\{\frac{|L\mathcal{D}|}{\ell^{2}N_{\ell}},\left(\frac{|L\mathcal{D}|}{\ell^{2}N_{\ell}}\right)^{2}\right\} \frac{E_{0}^{\rm GP}\left(K_{\ell},\varrho\ell^{2}\right)}{\ell^{2}}. 
		\end{equation}
		Notice that 
		$$\ell^{2}N_{\ell} = \big|\bigcup_{j=1}^{N_{\ell}} K_{j}\big| = (1+o(1)_{L \to \infty}) |L\mathcal{D}|,
		$$
		by \eqref{ineq-upper:thermodynamic-limit-general-1}, and $\ell = L^\eta \to \infty$. Thus, taking the limit $L\to\infty$ in \eqref{ineq-upper:thermodynamic-limit-general-2} and using Lemma~\ref{lem:thermodynamic-limit-square} we obtain the desired upper bound in \eqref{limit:thermodynamic}.
		
		\textbf{Lower bound.} We cover $L\mathcal{D}$ with squares $K_{j}$, $j=1, \ldots, N_{\ell}$ again, this time keeping the full covering but still having $\ell^{2} N_{\ell}|L\mathcal{D}|^{-1} \to 1$ as $L \to \infty$. Denote by $M_{\ell}$ the integer part of $N_{\ell}^{\frac{1}{2}}$, i.e.,
		$$
		M_{\ell} = \left\lfloor N_{\ell}^{\frac{1}{2}} \right\rfloor,
		$$
		where we used the notation $\lfloor x \rfloor = \max\{m\in \mathbb{Z} : m\leq x\}$. By definition, we have
		$$
		N_{\ell}^{\frac{1}{2}} - 1 \leq M_{\ell} \leq N_{\ell}^{\frac{1}{2}}.
		$$
		We pick any $M_{\ell}^{2}$ squares $K_{j}$, $j=1,\ldots,M_{\ell}^{2}$, among the $N_{\ell}$ squares. The area not covered said squares can be estimated as
		\begin{equation}\label{ineq-lower:thermodynamic-limit-general-1}
		\left|\left(\bigcup_{j=M_{\ell}^{2}+1}^{N_{\ell}} K_{j}\right)\right| = (N_{\ell} - M_{\ell}^{2})\ell^{2} \leq \left(N_{\ell}^{\frac{1}{2}} + M_{\ell} \right)\ell^{2} \leq 2N_{\ell}^{\frac{1}{2}}\ell^{2} = o(L^{2}).
		\end{equation}
		Next, we pick a minimizer $u^{\rm GP} = u_L^{\rm GP} \in H_{0}^{1}(L\mathcal{D})$ for $E_{0}^{\rm GP}(L\mathcal{D},\varrho |L\mathcal{D}|)$, and set
		$$
		\varrho_{j} := \frac{1}{\ell^{2}}\int_{K_{j}}\big|u^{\rm GP}\big|^{2}.
		$$
		Note that
		\begin{equation}\label{ineq-lower:thermodynamic-limit-general-2}
		\sum_{j=1}^{N_{\ell}}\varrho_{j}\ell^{2} = \int_{L\mathcal{D}}\big|u^{\rm GP}\big|^{2} = \varrho|L\mathcal{D}|
		\end{equation}
		and the mass concentrated outside $M_{\ell}^{2}$ squares is relatively small. Indeed, by Lemma~\ref{lem:thermodynamic-limit-boundedness} and \eqref{ineq-lower:thermodynamic-limit-general-1}, we have
		\begin{equation}\label{ineq-lower:thermodynamic-limit-general-3}
		\sum_{j=M_{\ell}^{2}+1}^{N_{\ell}}\varrho_{j}\ell^{2} \leq \left|\left(\bigcup_{j=M_{\ell}^{2}+1}^{N_{\ell}} K_{j}\right)\right|^{\frac{1}{2}} \left(\int_{L\mathcal{D}} \big|u^{\rm GP}\big|^{4}\right)^{\frac{1}{2}} \leq C N_{\ell}^{\frac{1}{4}}\ell L = o(L^{2}).
		\end{equation}
		Now we can estimate the energy. The idea of the proof is reminiscent of that in the upper bound part. We gauge away the rotation interaction between the $M_{\ell}^{2}$ squares, and this leads to a lower bound in terms of the Neumann energy in the square $K_{\ell M_{\ell}}$ of side length $\ell M_{\ell}$, centered at the origin. To bound the latter from below, we cover the square $K_{\ell M_{\ell}}$ with squares $\tilde{K}_{j}$, $j=1, \ldots, M_{\ell}^{2}$, of side length $\ell$, centered at $\tilde{x}_{j}$. We now estimate, using the gauge covariance of the functional on each $K_{j}$ and $\tilde{K}_{j}$,
		\begin{align}\label{ineq-lower:thermodynamic-limit-general-final-1}
		\mathcal{E}_{L\mathcal{D}}^{\rm GP}\big[u^{\rm GP}\big] \geq \sum_{j=1}^{M_{\ell}^{2}}\mathcal{E}_{K_{j}}^{\rm GP}\big[u^{\rm GP}\big]
		& = \sum_{j=1}^{M_{\ell}^{2}}\mathcal{E}_{\tilde{K}_{j}}^{\rm GP}\big[e^{\mathrm{i}\tilde{\phi}_{j}(\cdot-\tilde{\mathbf{x}}_{j})}e^{-\mathrm{i}\phi_{j}(\cdot+\mathbf{x}_{j}-\tilde{\mathbf{x}}_{j})}u^{\rm GP}(\cdot + \mathbf{x}_{j}-\tilde{x}_{j})\big] \nonumber \\
		& \geq \sum_{j=1}^{M_{\ell}^{2}}E^{\rm GP}\big(\tilde{K}_{j},\varrho_{j}\ell^{2}\big),
		\end{align}
		where $\phi_{j}$ and $\tilde{\phi}_{j}$ satisfy
		$$
		\begin{cases}
		\mathbf{x}^{\perp} - (\mathbf{x}-\mathbf{x}_{j})^{\perp}= \nabla \phi_{j} & \text{in} \quad  K_{j},\\
		(\mathbf{x}+\tilde{\mathbf{x}}_{j})^{\perp}- \mathbf{x}^{\perp} = \nabla \tilde{\phi}_{j} & \text{in} \quad  \tilde{K}_{j}.
		\end{cases}
		$$
		By \eqref{upper:estimate-DirNeu}, we have
		\begin{equation}\label{ineq-lower:thermodynamic-limit-general-final-2}
		E^{\rm GP}\big(\tilde{K}_{j},\varrho_{j}\ell^{2}\big) \geq E_{0}^{\rm GP}\big(\tilde{K}_{j},\varrho_{j}\ell^{2}\big) - C(1+G\varrho_{j})\left(\ell G^{-1} + \varrho_{j} \ell^{\frac{3}{2}}\right).
		\end{equation}
		Now we consider $\tilde{u}_{j}$, $j = 1,\ldots,M_{\ell}^{2}$, a minimizer for $E_{0}^{\rm GP}\big(\tilde{K}_{j},\varrho_{j}\ell^{2}\big)$. We use $\sum_{j=1}^{M_{\ell}^{2}}\tilde{u}_{j}$ as a trial state for the Dirichlet problem of mass $\sum_{j=1}^{M_{\ell}^{2}}\varrho_{j}\ell^{2}$ in a square $K_{\ell M_{\ell}}$ with side length $\ell M_{\ell}$ centered at the origin. We finally obtain from \eqref{ineq-lower:thermodynamic-limit-general-final-1} and \eqref{ineq-lower:thermodynamic-limit-general-final-2} that
		\begin{align}\label{ineq-lower:thermodynamic-limit-general-8}
		\frac{E_{0}^{\rm GP}(L\mathcal{D},\varrho |L\mathcal{D}|)}{|L\mathcal{D}|} & \geq \frac{E_{0}^{\rm GP}\left(K_{\ell M_{\ell}}, \sum_{j=1}^{M_{\ell}^{2}}\varrho_{j}\ell^{2}\right)}{|L\mathcal{D}|} - C\frac{\sum_{j=1}^{M_{\ell}^{2}}(1+G\varrho_{j})\left(\ell G^{-1} + \varrho_{j} \ell^{\frac{3}{2}}\right)}{|L\mathcal{D}|} \nonumber \\
		& \geq \frac{\ell^{2}M_{\ell}^{2}}{|L\mathcal{D}|} \min\left\{\frac{\sum_{j=1}^{M_{\ell}^{2}}\varrho_{j}\ell^{2}}{\varrho\ell^{2}M_{\ell}^{2}},\left(\frac{\sum_{j=1}^{M_{\ell}^{2}}\varrho_{j}\ell^{2}}{\varrho\ell^{2}M_{\ell}^{2}}\right)^{2}\right\}\frac{E_{0}^{\rm GP}\left(K_{\ell M_{\ell}}, \varrho\ell^{2}M_{\ell}^{2}\right)}{\ell^{2}M_{\ell}^{2}} \nonumber \\
		& \quad - C\frac{G^{-1}\ell M_{\ell}^{2} + \sum_{j=1}^{M_{\ell}^{2}}\varrho_{j}\ell + \varrho_{j}\ell^{\frac{3}{2}} + G\varrho_{j}^{2}\ell^{\frac{3}{2}}}{|L\mathcal{D}|}.
		\end{align}
		For the main term in \eqref{ineq-lower:thermodynamic-limit-general-8}, we have, by \eqref{ineq-lower:thermodynamic-limit-general-1}, \eqref{ineq-lower:thermodynamic-limit-general-2} and \eqref{ineq-lower:thermodynamic-limit-general-3},
		$$
		\ell^{2}M_{\ell}^{2} = \ell^{2}N_{\ell} + o(L^{2}) = |L\mathcal{D}| + o(L^{2})
		$$
		and
		$$
		\sum_{j=1}^{M_{\ell}^{2}}\varrho_{j}\ell^{2} = \varrho|L\mathcal{D}| + o(L^{2}).
		$$
		For the error term in \eqref{ineq-lower:thermodynamic-limit-general-8}, if we assume that $\ell = L^{\eta}$ with $\eta > \frac{4}{5}$ then
		$$
		\sum_{j=1}^{M_{\ell}^{2}}\varrho_{j}^{2}\ell^{\frac{3}{2}} = \ell^{-\frac{5}{2}}\sum_{j=1}^{M_{\ell}^{2}}\left(\varrho_{j}\ell^{2}\right)^{2} \leq \ell^{-\frac{5}{2}}\left(\sum_{j=1}^{M_{\ell}^{2}}\varrho_{j}\ell^{2}\right)^{2} \leq \ell^{-\frac{5}{2}}\left(\varrho|L\mathcal{D}|\right)^{2} = o(L^{2}).
		$$
		Note that $\ell M_{\ell} = L^\eta M_{\ell} \to \infty$. Thus, taking the limit $L\to\infty$ in \eqref{ineq-lower:thermodynamic-limit-general-8} and using Lemma~\ref{lem:thermodynamic-limit-square}, we obtain the desired lower bound in \eqref{limit:thermodynamic}.
	\end{proof}	
	
	\subsection{Low density regime}\label{sec:rho to 0}
	
	Now that we have proved that the thermodynamic limit of the homogeneous energy is the same with Neumann or Dirichlet conditions, it makes sense that the limit with periodic boundary conditions also coincides. Some care must be taken to define the latter, for the magnetic Laplacian does not commute with translations. The remedy is well-known (see e.g., \cite[Section~3.13]{Jain-07} or the discussion in \cite{Perice-22}): we impose so-called magnetic periodic boundary conditions on squares containing a quantized magnetic flux. The Abrikosov constant~\eqref{eq:def Ab} is best defined in terms of the low-density limit of the so-obtained problem, for there is then a well-defined, explicit, analogue \cite{AftSer-07,Almog-06,FouKac-11} of the lowest Landau level~\eqref{space:LLL}.
	
	Let $L > 0$ and denote by $K_{L}$ the unit square of the lattice $L(\mathbb{Z} \oplus i \mathbb{Z})$. We assume the quantization condition that $(2 \pi)^{-1}|K_{L}|$ is an integer, i.e., there exists $d \in \mathbb{N}$ such that
	\begin{equation}\label{space:LLL-dimension}
	L^{2}=2 \pi d.
	\end{equation}
	Let us introduce the following space
	\begin{equation}\label{space:LLL-square}
	\begin{aligned}
	H_{\rm per}^{1}(K_{L}) = \Big\{u \in H^{1}(K_{L}): u\left(x_{1}+L, x_{2}\right) & = e^{i \frac{\pi d x_{2}}{L}} u\left(x_{1}, x_{2}\right) \\ u\left(x_{1}, x_{2}+L\right) & = e^{-i \frac{\pi dx_{1}}{L}} u\left(x_{1},x_{2}\right)\Big\}.
	\end{aligned}
	\end{equation}
	The operator $\frac{1}{2}\big(-\mathrm{i}\nabla - \mathbf{x}^{\perp}\big)^{2}$ in $L^{2}(K_{L})$ is self-adjoint positive over the subspace $H_{\rm per}^{1}(K_{L})$. Properties of this operator were studied by Aftalion and Serfaty \cite{AftSer-07} (see also Almog \cite{Almog-06}). The following proposition is essentially \cite[Proposition 3.1]{AftSer-07}.
	
	\begin{proposition}[\textbf{Finite dimensional lowest Landau level}]\label{pro:LLL-finite}~\\
		Assume $L$ is such that $|K_{L}| \in 2\pi \mathbb{N}$. We have the following spectral properties:
		\begin{itemize}
			\item[(i)] The lowest eigenvalue of $\frac{1}{2}\big(-\mathrm{i}\nabla - \mathbf{x}^{\perp}\big)^{2}$ is equal to $1$, and the associated eigenspace, called $\mathcal{LLL}_{L}$, has complex dimension $d$ given by \eqref{space:LLL-dimension}.
			\item[(ii)] The second eigenvalue of $\frac{1}{2}\big(-\mathrm{i}\nabla - \mathbf{x}^{\perp}\big)^{2}$ is greater than $3$.
		\end{itemize}
	\end{proposition}
	
	The space $\mathcal{LLL}_{L}$ is the finite-dimensional analogue of the lowest Landau level in \eqref{space:LLL}. Let us now define the following energy with magnetic-periodic boundary conditions
	$$
	E_{\rm per}^{\rm GP}(\mathcal{D},M) = \inf\left\{\mathcal{E}_{\mathcal{D}}^{\rm GP}[u] : u\in H_{\rm per}^{1}(\mathcal{D}), \int_{\mathcal{D}}|u|^{2} = M\right\}.
	$$
	Since $H_{0}^{1}(K_{L})$ can be viewed as a subspace of $H_{\rm per}^{1}(K_{L})$ (conditions \eqref{space:LLL-square} are satisfied), we have
	\begin{equation}\label{DirNeuPer}
	E^{\rm GP}(K_{L},M) \leq E_{\rm per}^{\rm GP}(K_{L},M) \leq E_{0}^{\rm GP}(K_{L},M).
	\end{equation}
	Then Lemma~\ref{lem:NeuDir} implies that, for fixed $G > 0$ and $\varrho > 0$,
	\begin{equation}\label{limit:thermodynamic-periodic}
	\boxed{e^{\rm GP}(\varrho) = \lim_{L\to\infty}\frac{E_{\rm per}^{\rm GP}\left(K_{L},\varrho L^{2}\right)}{L^{2}}.}
	\end{equation}
	Using \eqref{limit:thermodynamic-periodic}, we derive an asymptotic formula for the thermodynamic limit $e^{\rm GP}(\varrho)$ as $\varrho \to 0$. This will be an important ingredient in the proof of our main result.
	
	\begin{theorem}[\textbf{Energy in the low density limit}]\label{thm:Abrikosov}~\\
		Let $G > 0$ be fixed and $\varrho \ll 1$. We have, as $L\to\infty$,
		\begin{equation}\label{scaling-law}
		\begin{aligned}
		\varrho L^{2} + (1 + o(1))\frac{e^{\rm Ab}(1)}{2}G\varrho^{2}L^{2} & \geq E_{\rm per}^{\rm GP}\left(K_{L},\varrho L^{2}\right) \\
		& \geq \varrho L^{2} + (1 + o(1))\frac{e^{\rm Ab}(1)}{2}\left(G\varrho^{2}-CG^{\frac{3}{2}}\varrho^{\frac{5}{2}}-CG^{2}\varrho^{3}\right)L^{2}
		\end{aligned}
		\end{equation}		
		for a constant $C > 0$. Here
		$$
		e^{\rm Ab}(\varrho) := \lim_{L\to\infty}\frac{1}{L^{2}} \inf\left\{\int_{K_{L}}|u|^{4} : u\in \mathcal{LLL}_{L}, \int_{K_{L}}|u|^{2} = \varrho L^{2}\right\}.
		$$
	\end{theorem}
	
	\begin{remark}[Thermodynamic limit at low density]~\\
		As a consequence of \eqref{limit:thermodynamic-periodic} and \eqref{scaling-law}, we have
		\begin{equation}\label{thermodynamic-limit:asymptotic}
		\frac{2}{G}\lim_{\varrho \to 0} \frac{e^{\rm GP}(\varrho) - \varrho}{\varrho^{2}} = e^{\rm Ab}(1).
		\end{equation}
		One can see immediately from the definition that $e^{\rm Ab}(1) \geq 1$. The minimization of $\fint_{K_{L}}|u|^{4}$ over $u \in \mathcal{LLL}_{L}$ is another formulation of the Abrikosov problem in finite domains \cite{AftBlaNie-06b,Almog-06}. Part of the proof of Theorem~\ref{thm:Abrikosov} is similar to that of \cite{AftBla-08} for the reduction of the Gross--Pitaevskii energy to the infinite-dimensional lowest Landau level. By using the Euler--Lagrange equation and elliptic estimates, we check that the periodic Gross--Pitaevskii minimizer and its projection onto the space $\mathcal{LLL}_{L}$ are close. In \cite{AftSer-07}, the projection onto the finite-dimensional lowest Landau level is also used, but with different elliptic estimates. By arguments similar to those in \cite{AftBla-08,AftSer-07}, an analogue of \eqref{thermodynamic-limit:asymptotic} is obtained in the regime
		\begin{equation}\label{rule-out:scaling}
		G^{\frac{1}{2}}\varrho^{\frac{1}{2}}L^{2} \to 0 \quad \text{where} \quad L\to\infty.
		\end{equation}
		We will not be at liberty to assume~\eqref{rule-out:scaling} when performing the local density approximation in the proof of our main theorem. In the following, we use elliptic estimates based on work by Fournais and Helffer \cite{FouHel-10} to circumvent the condition \eqref{rule-out:scaling}. This is reminiscent of considerations from \cite{FouKac-13},
		see in particular Theorem 2.12 and Remark 2.13 therein. \hfill$\diamond$
	\end{remark}
	%
	%
	%
	\begin{proof}[Proof of Theorem \ref{thm:Abrikosov}]
		Let $u$ be a minimizer for the variational problem
		$$
		E^{\rm Ab}\left(K_{L},\varrho L^{2}\right) := \inf\left\{\int_{K_{L}}|u|^{4} : u\in \mathcal{LLL}_{L}, \int_{K_{L}}|u|^{2} = \varrho L^{2}\right\}.
		$$
		By a simple scaling,
		$$
		e^{\rm Ab}(\varrho) = \lim_{L\to\infty}\frac{E^{\rm Ab}\left(K_{L},\varrho L^{2}\right)}{L^{2}} = e^{\rm Ab}(1)\varrho^{2}.
		$$
		By the variational principle, we have
		$$
		E_{\rm per}^{\rm GP}\left(K_{L},\varrho L^{2}\right) \leq \mathcal{E}_{K_{L}}^{\rm GP}[u] = \varrho L^{2} + \frac{e^{\rm Ab}(1)}{2}G\varrho^{2}L^{2}(1 + o(1)_{L\to\infty}).
		$$
		This is the desired upper bound in \eqref{scaling-law}.
		
		In order to obtain the lower bound in \eqref{scaling-law}, we denote by $u$ a minimizer for $E_{\rm per}^{\rm GP}\left(K_{L},\varrho L^{2}\right)$. Such $u$ solves the Ginzburg--Landau type equation
		\begin{equation}\label{eq:GL-bounded}
		\frac{1}{2}\big(-\mathrm{i}\nabla - \mathbf{x}^{\perp}\big)^{2}u + G|u|^{2}u = \lambda u \quad \text{in} \quad  K_{L},
		\end{equation}
		where $\lambda$ is the Euler--Lagrange multiplier. It follows from the above equation that
		\begin{equation}\label{eq:EL-GL}
		\lambda\int_{K_{L}}|u|^{2} =  \frac{1}{2}\int_{K_{L}} \big|\big(-\mathrm{i}\nabla - \mathbf{x}^{\perp}\big)u\big|^{2} + G\int_{K_{L}}|u|^{4}.
		\end{equation}
		The lowest eigenvalue of $\frac{1}{2}\big(-\mathrm{i}\nabla - \mathbf{x}^{\perp}\big)^{2}$ is equal to $1$, by Proposition \ref{pro:LLL-finite}.  We then infer from \eqref{eq:EL-GL} that $\lambda \geq 1$. On the other hand, it follows from \eqref{eq:EL-GL} and the upper bound on $E_{\rm per}^{\rm GP}\left(K_{L},\varrho L^{2}\right)$ in \eqref{scaling-law} that
		$$
		(\lambda + 1)\varrho L^{2} = (\lambda + 1)\int_{K_{L}}|u|^{2} \leq 2\mathcal{E}_{K_{L}}^{\rm GP}[u] \leq 2\varrho L^{2} + e^{\rm Ab}(1)G\varrho^{2}L^{2}(1 + o(1)_{L\to\infty}).
		$$
		This implies that
		\begin{equation}\label{ineq:projection-0}
		\lambda - 1 \leq e^{\rm Ab}(1)G\varrho(1 + o(1)_{L\to\infty}).
		\end{equation}
		Next, we define $v = \left(\frac{G}{\lambda}\right)^{\frac{1}{2}}u$. Then $v$ solves the equation
		\begin{equation}\label{eq:GL-bounded-fake}
		\frac{1}{2}\big(-\mathrm{i}\nabla - \mathbf{x}^{\perp}\big)^{2}v = \lambda(1-|v|^{2})v \quad \text{in} \quad  K_{L}.
		\end{equation}
		It follows from \cite[Theorem 3.1]{FouHel-10} and \eqref{ineq:projection-0} that
		\begin{equation}\label{ineq:elliptic-bounded}
		\|v\|_{L^{\infty}(K_{L})} \leq \min\left\{1,C_{\max}(\lambda-1)^{\frac{1}{2}}\right\},
		\end{equation}
		for a universal constant $C_{\max} > 0$. We remark that Fournais and Helffer \cite{FouHel-10} derived a uniform bound for Ginzburg--Landau type solutions on the whole space $\mathbb{R}^{2}$. This result is also true for ``periodic" solutions of the equation \eqref{eq:GL-bounded} in a bounded domain. Indeed, we tile the plane with squares $K_{j}$, $j=1,2,\ldots$, centered at $\mathbf{x}_{j}$ and of the side length $L$. We obtain from \eqref{eq:GL-bounded-fake} that
		$$
		\frac{1}{2}\big(-\mathrm{i}\nabla - \mathbf{x}^{\perp}\big)^{2}v_{j} = \lambda(1-|v_{j}|^{2})v_{j} \quad \text{in} \quad  K_{j}
		$$
		where $v_{j} = e^{\mathrm{i}\phi_{j}}v(\cdot+\mathbf{x}_{j})$ and the phase $\phi_{j}$ are chosen in such a way that
		$$
		\mathbf{x}^{\perp} - (\mathbf{x}+\mathbf{x}_{j})^{\perp}= \nabla \phi_{j} \quad \text{in} \quad  K_{j}.
		$$
		The function $v_{0} := \sum_{j}v_{j}\mathbbm{1}_{K_{j}}$ is a solution of the equation
		$$
		\frac{1}{2}\big(-\mathrm{i}\nabla - \mathbf{x}^{\perp}\big)^{2}v_{0} = \lambda(1-|v_{0}|^{2})v_{0} \quad \text{in} \quad  \mathbb{R}^{2}.
		$$
		Then \cite[Theorem 3.1]{FouHel-10} implies that
		$$
		\|v_{0}\|_{L^{\infty}(\mathbb{R}^{2})} \leq \min\left\{1,C_{\max}(\lambda-1)^{\frac{1}{2}}\right\},
		$$
		and hence \eqref{ineq:elliptic-bounded}. Now, \eqref{ineq:elliptic-bounded} and \eqref{ineq:projection-0} imply that
		\begin{equation}\label{ineq:projection-boundedness}
		\|u\|_{L^{\infty}(K_{L})} = \left(\frac{\lambda}{G}\right)^{\frac{1}{2}}\|v\|_{L^{\infty}(K_{L})} \leq C\left(\frac{\lambda}{G}\right)^{\frac{1}{2}}(\lambda-1)^{\frac{1}{2}} \leq C\varrho^{\frac{1}{2}}(1+o(1)_{L\to\infty}),
		\end{equation}
		for a universal constant $C$ independent of $L$. 
		
		Let $\Pi_L$ be the orthogonal projector on $\mathcal{LLL}_L$. We show in Appendix~\ref{app:boundedness-lll} that it is bounded  on $L^{2} \cap L^{\infty}(K_{L})$, independently of $L$. Hence 
		\begin{equation}\label{ineq:projection-1}
		\|\Pi_{L}u\|_{L^{\infty}(K_{L})} \leq C\|u\|_{L^{\infty}(K_{L})} \leq C\varrho^{\frac{1}{2}}(1+o(1)_{L\to\infty}),
		\end{equation} 
		follows from \eqref{ineq:projection-boundedness} for some constant $C$ independent of $L$. Let 
		$$
		\Pi_{L}^{\perp}u := u - \Pi_{L}u.
		$$
		Recall that the second eigenvalue of $\frac{1}{2}\big(-\mathrm{i}\nabla - \mathbf{x}^{\perp}\big)^{2}$ is at least $3$, by Proposition \ref{pro:LLL-finite}. Consequently,
		$$
		\frac{1}{2}\int_{K_{L}} \big|\big(-\mathrm{i}\nabla - \mathbf{x}^{\perp}\big)\Pi_{L}^{\perp}u\big|^{2} \geq 3\int_{K_{L}} \big|\Pi_{L}^{\perp}u\big|^{2}.
		$$
		Therefore,
		\begin{align*}
		E_{\rm per}^{\rm GP}\left(K_{L},\varrho L^{2}\right) = \mathcal{E}_{K_{L}}^{\rm GP}[u] & = \frac{1}{2}\int_{K_{L}} \big|\big(-\mathrm{i}\nabla - \mathbf{x}^{\perp}\big)\Pi_{L}u\big|^{2} + \big|\big(-\mathrm{i}\nabla - \mathbf{x}^{\perp}\big)\Pi_{L}^{\perp}u\big|^{2} + G|u|^{4} \\
		& \geq \int_{K_{L}} \big|\Pi_{L}u\big|^{2} + 3\big|\Pi_{L}^{\perp}u\big|^{2} + \frac{G}{2}|u|^{4} \\
		& \geq \varrho L^{2} + 2\int_{K_{L}} \big|\Pi_{L}^{\perp}u\big|^{2}.
		\end{align*}
		Then the upper bound on $E_{\rm per}^{\rm GP}\left(K_{L},\varrho L^{2}\right)$ in \eqref{scaling-law} implies that
		\begin{equation}\label{ineq:projection-2}
		\|\Pi_{L}^{\perp}u\|_{L^{2}(K_{L})} \leq CG^{\frac{1}{2}}\varrho L(1 + o(1)_{L\to\infty}).
		\end{equation}
		Next, we expand the quartic term of the energy as in \cite{AftBla-08}, and find
		\begin{align}
		\mathcal{E}_{K_{L}}^{\rm GP}[u] &= \mathcal{E}_{K_{L}}^{\rm GP}\big[\Pi_{L}u\big] + \mathcal{E}_{K_{L}}^{\rm GP}\big[\Pi_{L}^{\perp}u\big] - \frac{G}{2}\int_{K_{L}}\big|\Pi_{L}^{\perp}u\big|^{4} \nonumber \\ 
		& \quad + G\int_{K_{L}} \big|\Pi_{L}u\big|^{2}\big|\Pi_{L}^{\perp}u\big|^{2} + 2\left(\Re\left(\Pi_{L}u \overline{\Pi_{L}^{\perp}u}\right) + \frac{1}{2}\big|\Pi_{L}^{\perp}u\big|^{2}\right)^{2} + 2\big|\Pi_{L}u\big|^{2} \Re\left(\Pi_{L}u \overline{\Pi_{L}^{\perp}u}\right) \nonumber \\
		& \geq \mathcal{E}_{K_{L}}^{\rm GP}\big[\Pi_{L}u\big] + \mathcal{E}_{K_{L}}^{\rm GP}\big[\Pi_{L}^{\perp}u\big] - \frac{G}{2}\int_{K_{L}}\big|\Pi_{L}^{\perp}u\big|^{4} - 2G\int_{K_{L}}\big|\Pi_{L}u\big|^{3}\big|\Pi_{L}^{\perp}u\big|. \label{scaling-law:lower-bound}
		\end{align}
		On the one hand, it follows from H\"older' inequality and \eqref{ineq:projection-1}, \eqref{ineq:projection-2} that
		\begin{equation}\label{ineq:projection-3}
		G\int_{K_{L}}\big|\Pi_{L}u\big|^{3}\big|\Pi_{L}^{\perp}u\big| \leq G\|\Pi_{L}u\|_{L^{\infty}(K_{L})}^{2}\|\Pi_{L}u\|_{L^{2}(K_{L})}\|\Pi_{L}^{\perp}u\|_{L^{2}(K_{L})} \leq CG^{\frac{3}{2}}\varrho^{\frac{5}{2}}L^{2}(1+o(1)_{L\to\infty}).
		\end{equation}
		On the other hand, let $v = \Pi_{L}u\|\Pi_{L}u\|_{L^{2}(K_{L})}^{-1}L$. Since $v \in \mathcal{LLL}_{L}$ and $\|v\|_{L^{2}(K_{L})}^{2} = L^{2}$ we have
		\begin{equation}\label{ineq:projection-4}
		G\int_{K_{L}}\big|\Pi_{L}u\big|^{4} = G\frac{\|\Pi_{L}u\|_{L^{2}}^{4}}{L^{4}}\int_{K_{L}}|v|^{4} \geq e^{\rm Ab}(1)G\left(\varrho - CG\varrho^{2}\right)^{2}L^{2}(1+o(1)_{L\to\infty}),
		\end{equation}
		where we have used \eqref{ineq:projection-2}. Inserting \eqref{ineq:projection-3}, \eqref{ineq:projection-4} into \eqref{scaling-law:lower-bound} and using again Proposition \ref{pro:LLL-finite}, we thus obtain 
		$$
		\mathcal{E}_{K_{L}}^{\rm GP}[u] \geq \varrho L^{2} + \frac{e^{\rm Ab}(1)}{2}\left(G\varrho^{2}-CG^{\frac{3}{2}}\varrho^{\frac{5}{2}}-CG^{2}\varrho^{3}\right)L^{2}\left(1 + o(1)_{L\to\infty}\right).
		$$
		This is the desired lower bound in \eqref{scaling-law}.
	\end{proof}

	\section{Local density approximation}\label{sec:proof-main-result}
	
	In this section, we prove the energy convergence of $E_{\Omega}^{\rm GP}$ to $E_{\Omega}^{\rm TF}$ presented in Theorem~\ref{thm:main}. The asymptotic behavior of $E_{\Omega}^{\rm LLL}$ then follows from that of $E_{\Omega}^{\rm GP}$ and the comparison between the GP and LLL energies in Appendix~\ref{app:gp-lll}. We choose $G = G_{\Omega} = \left(1-\Omega^{2}\right)^{-\delta}$ with $-1 < \delta < 1$. In this case, we have, by \eqref{energy:TF-scaling} and \eqref{supp:TF-min},
	\begin{equation}\label{TF:energy-supp}
	E_{\Omega}^{\rm TF} \propto \left(1-\Omega^{2}\right)^{\frac{1-\delta}{2}} \quad \text{and} \quad L_{\Omega}^{\rm TF} = \left(1-\Omega^{2}\right)^{-\frac{1+\delta}{4}}.
	\end{equation}
	
	\subsection{Energy upper bound} Here we prove the upper bound corresponding to \eqref{convergence:energy-GP-TF}, i.e.,
	\begin{equation}\label{upper:local-density-approximation}
	E_{\Omega}^{\rm GP} - 1 \leq (1+o(1))E_{\Omega}^{\rm{TF}}.
	\end{equation}
	Let $\rho_{\Omega}^{\rm{TF}}$ be a minimizer for $E_{\Omega}^{\rm{TF}}$. We start by covering the support of $\rho_{\Omega}^{\rm{TF}}$ with squares $K_{j}$,  $j=1, \ldots, N_{L}$, centered at points $\mathbf{x}_{j}$ and of side length $L$ with
	\begin{equation}\label{length:square}
	L = \left(1-\Omega^{2}\right)^{-\eta} \quad \text{where} \quad 0 < \eta < \frac{1+\delta}{4}
	\end{equation}
	ensuring that 
	$$ 1 \ll L \ll \LTF.$$
	We choose the tiling in such a way that $K_{j} \cap \operatorname{supp}\left(\rho_{\Omega}^{\rm{TF}}\right) \neq \varnothing$, for any $j=1, \ldots, N_{L}$. The upper bound on $L$ indicates that the length scale of the tiling is much smaller than the size of the Thomas--Fermi support. Our trial state is defined much as in the proof of Lemmas \ref{lem:thermodynamic-limit-boundedness} and \ref{lem:thermodynamic-limit-square}:
	\begin{equation}\label{upper:trial-state}
	u^{\rm{test}} := \sum_{j=1}^{N_{L}} e^{\mathrm{i}\phi_{j}}u_{j}(\cdot - \mathbf{x}_{j}).
	\end{equation}
	Here $u_{j}$ realizes the Dirichlet infimum
	$$
	E_{0}^{\rm GP}\left(K_{L},\Omega,\varrho_{j}L^{2}\right) := \inf\left\{\mathcal{E}_{K_{L},\Omega}^{\rm GP}[u] : u\in H_{0}^{1}(K_{L}), \int_{K_{L}}|u|^{2} = \varrho_{j}L^{2}\right\}
	$$
	where
	$$
	\mathcal{E}_{K_{L},\Omega}^{\rm GP}[u] = \frac{1}{2}\int_{K_{L}} \big|\big(-\mathrm{i}\nabla - \mathbf{x}^{\perp}\big)u\big|^{2} + G_{\Omega}|u|^{4}
	$$
	and we set
	$$
	\varrho_{j} = \frac{1}{L^{2}}\int_{K_{L}}|u_{j}|^{2} := \frac{1}{L^{2}}\int_{K_{j}} \rho_{\Omega}^{\rm{TF}}.
	$$
	The phase factors in \eqref{upper:trial-state} are chosen in such a way that
	$$
	\mathbf{x}^{\perp} - (\mathbf{x}-\mathbf{x}_{j})^{\perp}= \nabla \phi_{j} \quad \text{in} \quad  K_{j}.
	$$ 
	This construction yields an admissible trial state since $u^{\rm{test}}$ is locally in $H^{1}(\mathbb{R}^{2})$, continuous across squares by being zero on the boundaries, and clearly
	$$
	\int_{\mathbb{R}^{2}}|u^{\rm{test}}|^{2} = \sum_{j=1}^{N_{L}} \int_{K_{L}}|u_{j}|^{2}=\sum_{j=1}^{N_{L}} \int_{K_{j}} \rho_{\Omega}^{\rm{TF}}=1.
	$$
	Much as in the proofs of Lemmas \ref{lem:thermodynamic-limit-boundedness} and \ref{lem:thermodynamic-limit-square} we thus obtain
	\begin{align}
	E_{\Omega}^{\rm GP} \leq \mathcal{E}_{\Omega}^{\rm GP}[u^{\rm{test}}] &= \sum_{j=1}^{N_{L}} \mathcal{E}_{K_{L},\Omega}^{\rm GP}[u_{j}] + \frac{1-\Omega^{2}}{2\Omega^{2}}\int_{\mathbb{R}^{2}} |\mathbf{x}|^{2}|u^{\rm{test}}|^{2} \nonumber \\
	& = \sum_{j=1}^{N_{L}} E_{0}^{\rm GP}\left(K_{L},\Omega,\varrho_{j}L^{2}\right) + \frac{1-\Omega^{2}}{2\Omega^{2}}\int_{\mathbb{R}^{2}} |\mathbf{x}|^{2}|u^{\rm{test}}|^{2}.\label{upper:estimate}
	\end{align}
	By \eqref{upper:estimate-DirNeu} we have
	$$
	E_{0}^{\rm GP}\left(K_{L},\Omega,\varrho_{j}L^{2}\right) \leq E^{\rm GP}\left(K_{L},\Omega,\varrho_{j}L^{2}\right) + C\left(LG_{\Omega}^{-1} + \varrho_{j}L^{\frac{3}{2}}\right).
	$$
	Therefore,
	\begin{equation}\label{upper:estimate-NeuDir-1}
	\sum_{j=1}^{N_{L}}E_{0}^{\rm GP}\left(K_{L},\Omega,\varrho_{j}L^{2}\right) \leq \sum_{j=1}^{N_{L}}E^{\rm GP}\left(K_{L},\Omega,\varrho_{j}L^{2}\right) + CN_{L}LG_{\Omega}^{-1} + CL^{-\frac{1}{2}}.
	\end{equation}
	By \eqref{TF:energy-supp} and \eqref{length:square}, the error term $N_{L}LG_{\Omega}^{-1}$ which is proportional to $\left(L_{\Omega}^{\rm TF}\right)^{2}L^{-1}G_{\Omega}^{-1}$ is of order $o\left(E_{\Omega}^{\rm TF}\right)$ when
	$$
	\eta > 1 - \delta.
	$$
	Together with the upper bound on $\eta$ in \eqref{length:square}, one needs
	\begin{equation}\label{condition:g}
	\boxed{\delta > \frac{3}{5}.}
	\end{equation}
	Furthermore, the error term $L^{-\frac{1}{2}}$ is also of order $o\left(E_{\Omega}^{\rm TF}\right)$ under the same condition \eqref{condition:g}. On the other hand, note that
	$$
	\varrho_{j} \leq L^{-2} = \left(1-\Omega^{2}\right)^{2\eta} \to 0
	$$	
	as $\Omega \nearrow 1$, uniformly with respect to $j=1,2,\ldots,N_L$. We thus deduce from \eqref{DirNeuPer} and \eqref{scaling-law} that
	$$
	E^{\rm GP}\left(K_{L},\Omega,\varrho_{j}L^{2}\right) \leq \varrho_{j}L^{2} + (1+o(1))\frac{e^{\rm Ab}(1)}{2}G_{\Omega}\varrho_{j}^{2} L^{2}.
	$$
	Therefore,
	\begin{equation}\label{upper:estimate-1}
	\sum_{j=1}^{N_{L}} E^{\rm GP}\left(K_{L},\Omega,\varrho_{j}L^{2}\right) \leq \sum_{j=1}^{N_{L}} \varrho_{j}L^{2} + (1+o(1)) \frac{e^{\rm Ab}(1)}{2}G_{\Omega}\sum_{j=1}^{N_{L}} \varrho_{j}^{2} L^{2}.
	\end{equation}
	By H\"older' inequality, we have
	\begin{equation}\label{upper:estimate-2}
	\sum_{j=1}^{N_{L}}\varrho_{j}^{2} L^{2} = \sum_{j=1}^{N_{L}}L^{-2}\left(\int_{K_{j}} \rho_{\Omega}^{\rm{TF}}\right)^{2} \leq \sum_{j=1}^{N_{L}}\int_{K_{j}}\left(\rho_{\Omega}^{\rm{TF}}\right)^{2} = \int_{\mathbb{R}^{2}}\left(\rho_{\Omega}^{\rm{TF}}\right)^{2}.
	\end{equation}
	Finally, we estimate the quadratic term in \eqref{upper:estimate}. We note that
	$$
	|\mathbf{x}| \leq CL_{\Omega}^{\rm TF}
	$$
	for any $\mathbf{x} \in \rm{supp}\left(\rho_{\Omega}^{\rm{TF}}\right)$, by \eqref{supp:TF-min}. Then
	\begin{align*}
	\int_{\mathbb{R}^{2}} |\mathbf{x}|^{2}|u^{\rm{test}}|^{2} & = \sum_{j=1}^{N_{L}} \int_{K_{j}} |\mathbf{x}|^{2}|u_{j}(\mathbf{x}-\mathbf{x}_{j})|^{2} \\
	& \leq \sum_{j=1}^{N_{L}} \int_{K_{j}} |\mathbf{x}_{j}|^{2}|u_{j}(\mathbf{x}-\mathbf{x}_{j})|^{2} + L\sum_{j=1}^{N_{L}} \int_{K_{j}} (|\mathbf{x}|+|\mathbf{x}_{j}|)|u_{j}(\mathbf{x}-\mathbf{x}_{j})|^{2} \\ 
	& \leq \sum_{j=1}^{N_{L}} \int_{K_{j}} |\mathbf{x}_{j}|^{2}\rho_{\Omega}^{\rm TF} + CLL_{\Omega}^{\rm TF}\sum_{j=1}^{N_{L}} \int_{K_{j}} |u_{j}(\mathbf{x}-\mathbf{x}_{j})|^{2} \\ 
	& \leq \sum_{j=1}^{N_{L}} \int_{K_{j}} |\mathbf{x}|^{2}\rho_{\Omega}^{\rm TF} + L\sum_{j=1}^{N_{L}} \int_{K_{j}} (|\mathbf{x}|+|\mathbf{x}_{j}|)\rho_{\Omega}^{\rm TF} + CLL_{\Omega}^{\rm TF}\\
	& \leq \int_{\mathbb{R}^{2}} |\mathbf{x}|^{2}\rho_{\Omega}^{\rm{TF}} + CLL_{\Omega}^{\rm TF}\sum_{j=1}^{N_{L}} \int_{K_{j}} \rho_{\Omega}^{\rm TF} + CLL_{\Omega}^{\rm TF}\\
	& \leq \int_{\mathbb{R}^{2}} |\mathbf{x}|^{2}\rho_{\Omega}^{\rm{TF}} + CLL_{\Omega}^{\rm TF}.
	\end{align*}
	The above implies that
	\begin{equation}\label{upper:estimate-3}
	\frac{1-\Omega^{2}}{\Omega^{2}}\int_{\mathbb{R}^{2}} |\mathbf{x}|^{2}|u^{\rm{test}}|^{2} \leq \frac{1-\Omega^{2}}{\Omega^{2}}\int_{\mathbb{R}^{2}} |\mathbf{x}|^{2}\rho_{\Omega}^{\rm{TF}} + CLG_{\Omega}^{\frac{1}{4}}\left(1-\Omega^{2}\right)^{\frac{3}{4}}.
	\end{equation}
	The error term $LG_{\Omega}^{\frac{1}{4}}\left(1-\Omega^{2}\right)^{\frac{3}{4}}$ is of order $ o\left(E_{\Omega}^{\rm{TF}}\right)$, by \eqref{TF:energy-supp} and \eqref{length:square}. The desired upper bound \eqref{upper:local-density-approximation} follows from \eqref{upper:estimate-1}, \eqref{upper:estimate-2} and \eqref{upper:estimate-3}.
	
	\subsection{Energy lower bound} Let us now complement \eqref{upper:local-density-approximation} by proving the lower bound
	\begin{equation}\label{lower:local-density-approximation}
	E_{\Omega}^{\rm GP} - 1 \geq (1+o(1))E_{\Omega}^{\rm{TF}}
	\end{equation}
	thus completing the proof of \eqref{convergence:energy-GP-TF}.
	
	Let $u^{\rm GP}$ be a minimizer for $E_{\Omega}^{\rm GP}$. The associated Euler--Lagrange equation takes the form
	\begin{equation}\label{eq:EL-GP}
	\frac{1}{2}\big(-\mathrm{i}\nabla - \mathbf{x}^{\perp}\big)^{2}u^{\rm GP} + G_{\Omega}\big|u^{\rm GP}\big|^{2}u^{\rm GP} + \frac{1-\Omega^{2}}{2\Omega^{2}}|\mathbf{x}|^{2}u^{\rm GP} = \lambda u^{\rm GP},
	\end{equation}
	with the Lagrange multiplier $\lambda$ given by
	$$
	\lambda = E_{\Omega}^{\rm GP} + \frac{G_{\Omega}}{2}\int_{\mathbb{R}^{2}}\big|u^{\rm GP}\big|^{4}.
	$$
	Using that $\lambda \geq 1$ we shall obtain uniform decay estimates \emph{\`a la} Agmon \cite{Agmon-82} for $u^{\rm GP}$. We will need the following exponential decay estimate in order to obtain that the GP mass outside the ball of TF radius $L_{\Omega}^{\rm TF}$ decays fast enough.
	
	\begin{lemma}[\textbf{Exponential decay of GP minimizers}]\label{lem:exponential-decay}~\\
		Let $u^{\rm GP}$ be $L^{2}$-normalized and solve~\eqref{eq:EL-GL} for some $\lambda >1$. There is a universal constant $C>0$ such that
		\begin{equation}\label{gp:decay-0}
		\int_{\mathbb{R}^{2}}e^{\sqrt{\lambda-1}|\mathbf{x}|}\big|u^{\rm GP}\big|^{2} \leq C\exp\left( C(\lambda-1)^{\frac{3}{2}}\left(1-\Omega^{2}\right)^{-\frac{1}{2}}\right).
		\end{equation}
	\end{lemma}
	
	\begin{proof}
		We use that (see e.g., \cite[Lemma 3.2]{LewNamRou-17proc})
		$$
		-\Re\left\langle u^{\rm GP},e^{\alpha |\mathbf{x}|}\Delta u^{\rm GP} \right\rangle = \int_{\mathbb{R}^{2}}\big|\nabla e^{\frac{\alpha}{2}|\mathbf{x}|}u^{\rm GP}\big|^{2} - \frac{\alpha^{2}}{4}\int_{\mathbb{R}^{2}}e^{\alpha |\mathbf{x}|}\big|u^{\rm GP}\big|^{2}.
		$$
		Then, we integrate the Euler--Lagrange equation \eqref{eq:EL-GP} against $e^{\alpha|\mathbf{x}|}\overline{u^{\rm GP}}$ and obtain
		$$
		\frac{1}{2}\int_{\mathbb{R}^{2}}\big|\big(-\mathrm{i}\nabla - \mathbf{x}^{\perp}\big)e^{\frac{\alpha}{2}|\mathbf{x}|}u^{\rm GP}\big|^{2} + 2G_{\Omega}e^{\alpha|\mathbf{x}|}\big|u^{\rm GP}\big|^{4} + \frac{1-\Omega^{2}}{\Omega^{2}}e^{\alpha|\mathbf{x}|}|\mathbf{x}|^{2}\big|u^{\rm GP}\big|^{2} = \left(\lambda + \frac{\alpha^{2}}{4}\right) \int_{\mathbb{R}^{2}}e^{\alpha |\mathbf{x}|}\big|u^{\rm GP}\big|^{2}.
		$$
		Using the operator inequality 
		\begin{equation}\label{ineq:operator}
		\frac{1}{2}(-\mathrm{i}\nabla - \mathbf{x}^\perp)^{2} \geq 1,
		\end{equation}
		choosing $\alpha = \sqrt{\lambda-1}$ and dropping the quartic term, we obtain that
		$$
		\frac{1-\Omega^{2}}{2\Omega^{2}}\int_{\mathbb{R}^{2}} e^{\sqrt{\lambda-1}|\mathbf{x}|}|\mathbf{x}|^{2}\big|u^{\rm GP}\big|^{2} \leq \frac{5}{4}(\lambda-1)\int_{\mathbb{R}^{2}}e^{\sqrt{\lambda-1}|\mathbf{x}|}\big|u^{\rm GP}\big|^{2}.
		$$
		Taking $R = C(\lambda - 1)^{\frac{1}{2}}(1 - \Omega^{2})^{-\frac{1}{2}}$, for some fixed large constant $C > 0$, and noticing that $u^{\rm GP}$ is of unit mass, the above inequality yields
		$$
		\int_{|\mathbf{x}| \geq R} e^{\sqrt{\lambda-1}|\mathbf{x}|}\big|u^{\rm GP}\big|^{2} \leq Ce^{\sqrt{\lambda-1}R},\quad \forall |\mathbf{x}| \geq R.
		$$
		This proves the desired exponential decay estimate.
	\end{proof}
	
	We now return to the energy lower bound \eqref{lower:local-density-approximation}. We again tile the plane with squares $K_{j}$,  $j=1, \ldots, N_{L}$, of side length $L = \left(1-\Omega^{2}\right)^{-\eta}$ satisfying
	\begin{equation}\label{length:square-enhance}
	1 - \delta < \eta < \frac{1+\delta}{4},
	\end{equation}
	and taken to cover the finite disk $B_{CL_{\Omega}^{\rm TF}}(0)$ (the support of $\rho_{\Omega}^{\rm{TF}}$). Let
	\begin{equation}\label{TF:appro-density}
	\varrho_{j} := \frac{1}{L^{2}}\int_{K_{j}}\big|u^{\rm GP}\big|^{2}.
	\end{equation}
	Define the piecewise constant function
	\begin{equation}\label{TF:appro-minimizer-GP}
	\overline{\rho}_{\Omega}^{\rm GP}:=\sum_{j=1}^{N_{L}} \varrho_{j} \mathbbm{1}_{K_{j}}.
	\end{equation}
	We first claim that we have, as $\Omega \nearrow 1$,
	\begin{equation}\label{gp:decay}
	\sum_{j=1}^{N_{L}} \varrho_{j}L^{2} = \int_{\mathbb{R}^{2}}\overline{\rho}_{\Omega}^{\rm GP} = 1 - o\left(E_{\Omega}^{\rm TF}\right). 
	\end{equation}
	Indeed, using the exponential decay of the GP minimizer in \eqref{gp:decay-0} we obtain
	\begin{equation}\label{gp:decay-1}
	\int_{B_{L_{\Omega}^{\rm TF}}^{c}(0)}\big|u^{\rm GP}\big|^{2} \leq C\exp\left(-(\lambda-1)^{\frac{1}{2}}L_{\Omega}^{\rm TF} + C(\lambda-1)^{\frac{3}{2}}\left(1-\Omega^{2}\right)^{-\frac{1}{2}}\right).
	\end{equation}
	From the operator inequality \eqref{ineq:operator} and the upper bound on $E_{\Omega}^{\rm GP}$ in \eqref{upper:local-density-approximation}, we have
	\begin{equation}\label{GL:bound-eigenvalue}
	\lambda - 1 \sim G_{\Omega}^{\frac{1}{2}}\left(1-\Omega^{2}\right)^{\frac{1}{2}} = \left(1-\Omega^{2}\right)^{\frac{1-\delta}{2}},
	\end{equation}
	with $\frac{3}{5} < \delta < 1$. Then one can easily check that, as $\Omega \nearrow 1$, the right hand side of \eqref{gp:decay-1} is extremely small. Hence
	$$
	\int_{B_{L_{\Omega}^{\rm TF}}^{c}(0)}\big|u^{\rm GP}\big|^{2} = o\left(E_{\Omega}^{\rm TF}\right),
	$$
	which yields \eqref{gp:decay}.
	
	Now we can estimate the energy. Dropping some positive terms we get
	\begin{align}
	E_{\Omega}^{\rm GP}=\mathcal{E}_{\Omega}^{\rm GP}\big[u^{\rm GP}\big] & \geq \sum_{j=1}^{N_{L}}\mathcal{E}_{K_{j},\Omega}^{\rm GP}\big[u^{\rm GP}\big] + \frac{1-\Omega^{2}}{2\Omega^{2}}\int_{\mathbb{R}^{2}} |\mathbf{x}|^{2}\big|u^{\rm GP}\big|^{2} \nonumber\\
	& = \sum_{j=1}^{N_{L}}\mathcal{E}_{K_{L},\Omega}^{\rm GP}\big[e^{-\mathrm{i}\phi_{j}(\cdot+\mathbf{x}_{j})}u^{\rm GP}(\cdot+\mathbf{x}_{j})\big] + \frac{1-\Omega^{2}}{2\Omega^{2}}\int_{\mathbb{R}^{2}} |\mathbf{x}|^{2}\big|u^{\rm GP}\big|^{2} \nonumber \\
	& \geq \sum_{j=1}^{N_{L}}E^{\rm GP}\left(K_{L},\Omega,\varrho_{j}L^{2}\right) + \frac{1-\Omega^{2}}{2\Omega^{2}}\int_{\mathbb{R}^{2}} |\mathbf{x}|^{2}\big|u^{\rm GP}\big|^{2}.\label{energy:lower-bound-1}
	\end{align}
	Here the local gauge phase factors are defined as in previous arguments by demanding that
	$$
	\mathbf{x}^{\perp} - (\mathbf{x}-\mathbf{x}_{j})^{\perp}= \nabla \phi_{j} \quad \text{in} \quad  K_{j}.
	$$
	On one hand, we deduce from \eqref{upper:estimate-NeuDir-1} that
	\begin{equation}\label{energy:lower-bound-2}
	\sum_{j=1}^{N_{L}}E^{\rm GP}\left(K_{L},\Omega,\varrho_{j}L^{2}\right) \geq \sum_{j=1}^{N_{L}}E_{0}^{\rm GP}\left(K_{L},\Omega,\varrho_{j}L^{2}\right) - o\left(E_{\Omega}^{\rm TF}\right),
	\end{equation}
	provided that \eqref{length:square-enhance} is satisfied. Furthermore, note that
	$$
	\varrho_{j} \leq L^{-2} = \left(1-\Omega^{2}\right)^{2\eta} \to 0.
	$$	
	as $\Omega \nearrow 1$, uniformly with respect to $j=1,2,\ldots,N_L$. We thus deduce from \eqref{DirNeuPer} and \eqref{scaling-law} that
	$$E_{0}^{\rm GP}\left(K_{L},\Omega,\varrho_{j}L^{2}\right) \geq \varrho_{j}L^{2} + (1+o(1)) \frac{e^{\rm Ab}(1)}{2}\left(G_{\Omega}\varrho_{j}^{2} - CG_{\Omega}^{\frac{3}{2}}\varrho_{j}^{\frac{5}{2}} - CG_{\Omega}^{2}\varrho_{j}^{3}\right)L^{2}.
	$$
	Note that \eqref{scaling-law} was proved when $G > 0$ is fixed. But the same result is obtained when $G = G_{\Omega} = \left(1-\Omega^{2}\right)^{-\delta}$ with $-1<\delta<1$. The arguments are the same by noting that $G_{\Omega}\varrho_{j} \leq G_{\Omega}L^{-2} \ll 1$. We deduce from the above that
	\begin{equation}\label{energy:lower-bound-3}
	\sum_{j=1}^{N_{L}}E_{0}^{\rm GP}\left(K_{L},\Omega,\varrho_{j}L^{2}\right) \geq \sum_{j=1}^{N_{L}} \varrho_{j}L^{2} + (1+o(1)) \frac{e^{\rm Ab}(1)}{2}\sum_{j=1}^{N_{L}} \left(G_{\Omega}\varrho_{j}^{2} - CG_{\Omega}^{\frac{3}{2}}\varrho_{j}^{\frac{5}{2}} - CG_{\Omega}^{2}\varrho_{j}^{3}\right)L^{2}.
	\end{equation}
	The error terms in \eqref{energy:lower-bound-3} can be estimated as follows
	$$
	\sum_{j=1}^{N_{L}}\left(G_{\Omega}^{\frac{3}{2}}\varrho_{j}^{\frac{5}{2}} + G_{\Omega}^{2}\varrho_{j}^{3}\right)L^{2} \leq \left(G_{\Omega}^{\frac{3}{2}}L^{-3} + G_{\Omega}^{2}L^{-4}\right) \sum_{j=1}^{N_{L}}\varrho_{j}L^{2} \leq  G_{\Omega}^{\frac{3}{2}}L^{-3} + G_{\Omega}^{2}L^{-4}.
	$$
	They are of order $o\left(E_{\Omega}^{\rm TF}\right)$ when
	$$
	\eta > \max\left\{\frac{2\delta+1}{6},\frac{3\delta+1}{8}\right\} = \frac{2\delta+1}{6}.
	$$
	Combining~\eqref{energy:lower-bound-2}, \eqref{energy:lower-bound-3} and using \eqref{gp:decay} we obtain
	\begin{equation}\label{energy:lower-bound-4}
	\sum_{j=1}^{N_{L}}E^{\rm GP}\left(K_{L},\Omega,\varrho_{j}L^{2}\right) \geq 1 + (1+o(1)) \frac{e^{\rm Ab}(1)}{2}G_{\Omega}\sum_{j=1}^{N_{L}} \varrho_{j}^{2}L^{2} + o\left(E_{\Omega}^{\rm{TF}}\right).
	\end{equation}
	On the other hand, let
	$$
	V_{1}(\mathbf{x}):=\sum_{j=1}^{N_{L}} |\mathbf{x}_{j}|^{2} \mathbbm{1}_{K_{j}}(\mathbf{x}) \quad \text{and} \quad V_{2}(\mathbf{x}):=\sum_{j=1}^{N_{L}} |\mathbf{x}_{j}| \mathbbm{1}_{K_{j}}(\mathbf{x}).
	$$
	Then we have, by the triangle inequality,
	\begin{align}
	\int_{\mathbb{R}^{2}} |\mathbf{x}|^{2}\big|u^{\rm GP}\big|^{2} \geq \sum_{j=1}^{N_{L}} \int_{K_{j}} |\mathbf{x}|^{2}\big|u^{\rm GP}\big|^{2} 
	& \geq \sum_{j=1}^{N_{L}} \int_{K_{j}} |\mathbf{x}_{j}|^{2}\big|u^{\rm GP}\big|^{2} - L\sum_{j=1}^{N_{L}} \int_{K_{j}} (|\mathbf{x}|+|\mathbf{x}_{j}|)\big|u^{\rm GP}\big|^{2}  \nonumber \\
	& \geq \int_{\mathbb{R}^{2}} V_{1} \overline{\rho}_{\Omega}^{\rm GP} - L\sum_{j=1}^{N_{L}} \int_{K_{j}} (2|\mathbf{x}|+L) \big|u^{\rm GP}\big|^{2} \nonumber \\
	& = \int_{\mathbb{R}^{2}} V_{1} \overline{\rho}_{\Omega}^{\rm GP} - 2L\int_{\mathbb{R}^{2}} |\mathbf{x}|\big|u^{\rm GP}\big|^{2} - L^{2}. \label{energy:lower-bound-5}
	\end{align}
	In the very same way however we can put back $|\mathbf{x}|^{2}$ in place of $V_{1}(\mathbf{x})$, obtaining
	\begin{align}
	\int_{\mathbb{R}^{2}} V_{1} \overline{\rho}_{\Omega}^{\rm GP} = \sum_{j=1}^{N_{L}} \int_{K_{j}} |\mathbf{x}_{j}|^{2}\varrho_{j}
	& \geq \sum_{j=1}^{N_{L}} \int_{K_{j}} |\mathbf{x}|^{2}\varrho_{j}  - L\sum_{j=1}^{N_{L}} \int_{K_{j}} (|\mathbf{x}|+|\mathbf{x}_{j}|)\varrho_{j}  \nonumber \\
	& \geq \int_{\mathbb{R}^{2}} |\mathbf{x}|^{2}\overline{\rho}_{\Omega}^{\rm GP} - \sum_{j=1}^{N_{L}} L\int_{K_{j}} (2|\mathbf{x}|+L) \varrho_{j} \nonumber \\
	& = \int_{\mathbb{R}^{2}} |\mathbf{x}|^{2}\overline{\rho}_{\Omega}^{\rm GP} - 2L\int_{\mathbb{R}^{2}} |\mathbf{x}|\overline{\rho}_{\Omega}^{\rm GP} - L^{2}. \label{energy:lower-bound-6}
	\end{align}
	Furthermore,
	\begin{align}
	\int_{\mathbb{R}^{2}} |\mathbf{x}| \overline{\rho}_{\Omega}^{\rm GP} = \sum_{j=1}^{N_{L}} \int_{K_{j}} |\mathbf{x}|\varrho_{j}
	& \leq \sum_{j=1}^{N_{L}} \int_{K_{j}} |\mathbf{x}_{j}|\varrho_{j} + L\sum_{j=1}^{N_{L}} \int_{K_{j}}\varrho_{j}  \nonumber \\
	& = \sum_{j=1}^{N_{L}} \int_{K_{j}} |\mathbf{x}_{j}|\big|u^{\rm GP}\big|^{2} + L \nonumber \\
	& \leq \sum_{j=1}^{N_{L}} \int_{K_{j}} |\mathbf{x}|\big|u^{\rm GP}\big|^{2} + L\sum_{j=1}^{N_{L}} \int_{K_{j}}\big|u^{\rm GP}\big|^{2}  + L \nonumber \\
	& = \int_{\mathbb{R}^{2}} |\mathbf{x}|\big|u^{\rm GP}\big|^{2} + 2L. \label{energy:lower-bound-7}
	\end{align}
	Combining \eqref{energy:lower-bound-5}, \eqref{energy:lower-bound-6} and \eqref{energy:lower-bound-7} yields
	\begin{align}
	\frac{1-\Omega^{2}}{\Omega^{2}}\int_{\mathbb{R}^{2}} |\mathbf{x}|^{2}\big|u^{\rm GP}\big|^{2} & \geq \frac{1-\Omega^{2}}{\Omega^{2}}\left(\int_{\mathbb{R}^{2}} |\mathbf{x}|^{2}\overline{\rho}_{\Omega}^{\rm GP} - 4L\int_{\mathbb{R}^{2}} |\mathbf{x}|\big|u^{\rm GP}\big|^{2} - 6L^{2}\right) \nonumber\\
	& = \frac{1-\Omega^{2}}{\Omega^{2}}\int_{\mathbb{R}^{2}} |\mathbf{x}|^{2}\overline{\rho}_{\Omega}^{\rm GP} + o\left(E_{\Omega}^{\rm TF}\right). \label{energy:lower-bound-8}
	\end{align}
	The last assertion follows from \eqref{TF:energy-supp} and \eqref{length:square}. Indeed, by \eqref{TF:energy-supp}, Holder's inequality and the upper bound on $E_{\Omega}^{\rm GP}$ in \eqref{upper:local-density-approximation}, we have
	$$
	\int_{\mathbb{R}^{2}} |\mathbf{x}|\big|u^{\rm GP}\big|^{2} \leq \left(\int_{\mathbb{R}^{2}} |\mathbf{x}|^{2}\big|u^{\rm GP}\big|^{2}\right)^{\frac{1}{2}}\left(\int_{\mathbb{R}^{2}} \big|u^{\rm GP}\big|^{2}\right)^{\frac{1}{2}} \leq C\left(1-\Omega^{2}\right)^{-\frac{1}{4}}G_{\Omega}^{\frac{1}{4}}.
	$$
	
	Putting together \eqref{energy:lower-bound-4}, \eqref{energy:lower-bound-8} and using \eqref{gp:decay}, one obtains
	\begin{align}
	E_{\Omega}^{\rm GP} & \geq 1 + (1+o(1)) \frac{e^{\rm Ab}(1)}{2}G_{\Omega}\int_{\mathbb{R}^{2}}\left(\overline{\rho}_{\Omega}^{\rm GP}\right)^{2} + \frac{1-\Omega^{2}}{2\Omega^{2}}\int_{\mathbb{R}^{2}} |\mathbf{x}|^{2}\overline{\rho}_{\Omega}^{\rm GP} + o\left(E_{\Omega}^{\rm{TF}}\right) \nonumber \\
	& \geq 1 + (1+o(1)) \mathcal{E}_{\Omega}^{\rm{TF}}\big[\overline{\rho}_{\Omega}^{\rm GP}\big] + o\left(E_{\Omega}^{\rm{TF}}\right) \nonumber \\
	& \geq 1 + (1+o(1)) \min\left\{\int_{\mathbb{R}^{2}}\overline{\rho}_{\Omega}^{\rm GP},\left(\int_{\mathbb{R}^{2}}\overline{\rho}_{\Omega}^{\rm GP}\right)^{2}\right\}\mathcal{E}_{\Omega}^{\rm{TF}}\left[\frac{\overline{\rho}_{\Omega}^{\rm GP}}{\int_{\mathbb{R}^{2}}\overline{\rho}_{\Omega}^{\rm GP}}\right] + o\left(E_{\Omega}^{\rm{TF}}\right) \nonumber \\
	& \geq 1 +  (1+o(1)) E_{\Omega}^{\rm{TF}}. \label{lower:local-density-approximation-1}
	\end{align}
	This completes the proof of \eqref{lower:local-density-approximation}.
	
	\begin{remark}[Limitations of the local density approximation]\label{rem:LDA}~\\
		We now explain why the limitation $G\gg(1-\Omega)^{-\frac{3}{5}}$ is necessary with our scheme of proof. We used repeatedly comparisons between Dirichlet and Neumann energies in squares of side length $L \ll L_{\Omega}^{\rm TF}$. As per the considerations of Remark~\ref{rem:NeuDir}, this brings about an error of order $L G^{-1}$ per square. Summed over all squares the total error cannot be less than 
		$$\frac{L}{G}\left(\frac{L_{\Omega}^{\rm TF}}{L}\right)^{2} \gg \frac{L_{\Omega}^{\rm TF}}{G}.$$
		Recalling that $E_{\Omega}^{\rm TF} \propto G^{\frac{1}{2}} (1-\Omega)^{\frac{1}{2}} $ while $L_{\Omega}^{\rm TF} = G^{\frac{1}{4}} (1-\Omega) ^{-\frac{1}{4}}$ (see~\eqref{supp:TF-min}), the error being smaller than the main term requires $G\gg(1-\Omega)^{-\frac{3}{5}}$.
		
		Thus, extending the validity of the local density approximation of Theorem~\ref{thm:main} to smaller values of $G$ seems to require another idea than Dirichlet--Neumann bracketing, as presented e.g. in Section~\ref{sec:reduc}. \hfill $\diamond$
	\end{remark}

	\subsection{Density convergence}
	
	We turn to the proof of~\eqref{convergence:density} under the assumptions of Theorem~\ref{thm:main}. We use a Feynman-Hellmann-Griffith argument, starting with perturbed functionals. We denote 
	$$ h = 1-\Omega^2, \quad L = \LTF$$
	for short and recall that we set $G= h^{-\delta}$. Let $\varphi$ be as in the statement of the theorem and $\eps > 0$. We set    
	\begin{equation}\label{eq:GPeps}
	\cGGP_\eps[u] = \frac{1}{2}\int_{\mathbb{R}^{2}} \big|\big(-\mathrm{i}\nabla - \mathbf{x}^{\perp}\big)u\big|^{2} + G|u|^{4} + \frac{1-\Omega^{2}}{\Omega^{2}}|\mathbf{x}|^{2}|u|^{2} + \eps (Gh)^{\frac{1}{2}} \varphi\left(\frac{\bx}{L}\right)|u|^{2}
	\end{equation}
	with the corresponding infimum
	\begin{equation}\label{eq:GPeps inf}
	\GGP_\eps = \inf\left\{\cGGP_\eps[u] : u\in H^{1}(\mathbb{R}^{2}), \int_{\mathbb{R}^{2}}|u|^{2} = 1\right\}.
	\end{equation}
	Similarly, let 
	\begin{equation}\label{eq:LLLeps}
	\cGLLL_\eps [u] := 1 + \frac{1}{2}\int_{\mathbb{R}^{2}} G|u|^{4} + \frac{1-\Omega^{2}}{\Omega^{2}}|\mathbf{x}|^{2}|u|^{2} + \eps (Gh)^{\frac{1}{2}} \varphi\left(\frac{\bx}{L}\right)|u|^{2}.
	\end{equation}
	with
	\begin{equation}\label{eq:LLLeps ing}
	\GLLL_\eps = \inf\left\{\cGLLL_\eps [u] : u\in \mathcal{LLL}, \int_{\mathbb{R}^{2}}|u|^{2} = 1\right\}.
	\end{equation}
	Our target functional is likewise
	\begin{equation}\label{eq:TFeps}
	\cGTF_\eps[\rho] := \frac{1}{2}\int_{\mathbb{R}^{2}}e^{\rm Ab}(1)G\rho^{2} + \frac{1-\Omega^{2}}{\Omega^{2}}|\mathbf{x}|^{2}\rho + \eps (Gh)^{\frac{1}{2}} \varphi\left(\frac{\bx}{L}\right)\rho
	\end{equation}
	with
	\begin{equation}\label{eq:TFeps inf}
	\GTF_\eps := \inf\left\{ \cGTF_\eps [\rho], \; \rho \in L^{1}\cap L^{2}(\mathbb{R}^{2} ; \mathbb{R}^{+}), \int_{\R^2} \rho = 1 \right\}.
	\end{equation}
	We denote $\rhoTF_\eps$ the unique minimizer of the above, that one can explicitly compute using the associated variational equation, similarly to~\eqref{eq:TF min}.
	
	We recall that with the above notation
	\begin{equation}\label{eq:scale ener} 
	\ETF = \GTF_{\eps = 0} \sim (Gh)^{\frac{1}{2}}.
	\end{equation}
	Since $\varphi$ is assumed to be Lipschitz, if $\eps$ is small enough (but fixed in the limit $h\to 0$) we can follow the arguments above with minor modifications to obtain
	\begin{equation}\label{eq:GP TF eps}
	\begin{aligned}
	 \GGP_\eps &= 1 + \GTF_\eps + o_h\left(\GTF_\eps\right) \\
	\GLLL_\eps &= 1 + \GTF_\eps + o_h\left(\GTF_\eps\right)
	\end{aligned}
	\end{equation}
    with $o_h(1) \to 0$ when $h\to 0$. 
    
    The argument essentially consists in differentiating (with care)~\eqref{eq:GP TF eps} with respect to $\eps$ and evaluating the result at $\eps = 0$. It is exactly similar in the GP and LLL cases, we thus discuss only the former. Let $\uGP$ satisfy~\eqref{eq:quasimin}, i.e. 
    $$ \cGGP_0 \left[\uGP\right] = \GGP_0 + o_h\left(\GTF_0\right).$$
    Then using the variational principles and~\eqref{eq:GP TF eps}, we have
	\begin{align*}
	  \eps (Gh)^{\frac{1}{2}} \int_{\R^2} \varphi\left(\frac{\bx}{L}\right)\left|\uGP\right|^{2} &= \cGGP_\eps \left[\uGP\right] - \cGGP_0 \left[\uGP\right] \\
	  & = \cGGP_\eps \left[\uGP\right] - \GGP_0 + o_h\left(\GTF_0\right) \\
	  &\geq \GGP_\eps - \GGP_0 + o_h\left(\GTF_0\right) \\
	  & = \GTF_\eps - \GTF_0 + o_h\left(\GTF_0\right)\\
	  &\geq \eps (Gh)^{\frac{1}{2}} \int_{\R^2} \varphi\left(\frac{\bx}{L}\right)\rhoTF_\eps + o_h\left(\GTF_0\right).
	\end{align*}
    Hence, recalling~\eqref{eq:scale ener},
	$$
	\int_{\R^2} \varphi\left(\frac{\bx}{L}\right)|\uGP|^{2} \geq \int_{\R^2} \varphi\left(\frac{\bx}{L}\right)\rhoTF_\eps + \eps^{-1}f(h) \geq \int_{\R^2} \varphi\left(\frac{\bx}{L}\right)\rhoTF_0 + \eps^{-1} f(h ) + g(\eps)
	$$
    where $f(h) \to 0$ when $h \to 0$ and $g(\eps) \to 0$ when $\eps \to 0$. The last estimate in the above follows from explicit calculations within TF theory, see e.g.~\cite[Equations~(2.13) to~(2.15)]{CorLunRou-16}. Changing variables, we find  
    $$
		 \int_{\R^2} \varphi(\bx) \left| \LTF\uGP \left(\LTF \bx\right) \right|^2  {\rm d}\bx \geq \int_{\R^2} \varphi \rho_{1}^{\rm TF} + \eps^{-1} f(h) + g(\eps)
		 $$
	where $\rho_{1}^{\rm TF}$ is defined as the minimizer of~\eqref{energy:TF-scaling}. We can then choose for example $\eps = \eps (h) = \sqrt{f(h)}\to 0$ when $h\to 0$ so that $\eps^{-1} f(h) \to 0$ in the above. This gives
	$$
	\liminf_{h\to 0} \int_{\R^2} \varphi(\bx) \left| \LTF\uGP \left(\LTF \bx\right) \right|^2 {\rm d}\bx \geq \int_{\R^2} \varphi \rho_{1}^{\rm TF}.
	$$
	Repeating the same argument with $\eps$ replaced by $-\eps$ gives  
	$$
	\limsup_{h\to 0} \int_{\R^2} \varphi(\bx) \left| \LTF\uGP \left(\LTF \bx\right) \right|^2 {\rm d}\bx \leq \int_{\R^2} \varphi \rho_{1}^{\rm TF}
	$$
	and thus concludes the proof of~\eqref{convergence:density}.

	\bigskip
	
	\appendix
	
	\section{Projector onto the finite-dimensional lowest Landau level}\label{app:boundedness-lll}
	
	In this appendix we prove the uniform boundedness of the projector $\Pi_{L}$ onto the lowest Landau level of finite-dimensional $\mathcal{LLL}_{L}$ in $L^{2}(K_{L})$. The projector $\Pi_{L}$ is constructed as a linear combination of  orthonormal projections on basis functions of the lowest Landau level. A convenient basis can be defined using Theta functions (see e.g., \cite[Proposition 3.1]{AftSer-07}). The kernel of $\Pi_{L}$ is found to be \cite{Perice-22} 
	$$
	\Pi_{L}(\mathbf{x}, \mathbf{y}) = \frac{2\sqrt{\pi}}{L^3} \sum_{l=0}^{d-1} \sum_{k, p \in \mathbb{Z}} e^{2\mathrm{i}\pi l \frac{p-k}{d}-2\mathrm{i}\pi \frac{kx_{2}-py_{2}}{L}-\frac{1}{2}\left(x_{1}+k \frac{L}{d}\right)^{2}-\frac{1}{2}\left(y_{1}+p \frac{L}{d}\right)^{2}}
	$$
	where $\mathbf{x}=x_{1}+ix_{2}$, $\mathbf{y}=y_{1}+iy_{2}$ and the integer $d$ is the dimension of $\mathcal{LLL}_{L}$ given by the quantization \eqref{space:LLL-dimension}, i.e., $2\pi d = L^{2}$.
	
	We need the following lemma in order to establish elliptic estimates for ``periodic" solutions of Ginzburg--Landau type equation \eqref{eq:GL-bounded-fake} on the bounded domain $K_{L}$.
	
	%
	%
	
	\begin{lemma}[\textbf{$L^p$ bounds for the finite-volume $\mathcal{LLL}$ projector}]\mbox{}\\
		Let $\Pi_{L}$ be the projection on the lowest Landau level $\mathcal{LLL}_{L}$ in $L^{2}(K_{L})$. For every function $u$ in $L^{p}(K_{L})$, there exists a universal constant $C > 0$, independent of $L$, such that
		$$
		\|\Pi_{L}u\|_{L^{p}(K_{L})} \leq C\|u\|_{L^{p}(K_{L})},
		$$
		whenever $2\leq p \leq \infty$.
	\end{lemma}
	
	\begin{proof}

		Obviously, $\Pi_{L}$ is a bounded operator on $L^{2}(K_{L})$. Here we prove that it is bounded on $L^{\infty}(K_{L})$. This yields the continuation of $\Pi_{L}$ on $L^{p}(K_{L})$ for all $p \in [2, \infty]$, by interpolation. For any $u\in L^{\infty}(K_{L})$ and $\mathbf{x} \in K_{L}$, we have
		$$
		|(\Pi_{L}u)(\mathbf{x})| = \left|\int_{K_{L}}\Pi_{L}(\mathbf{x},\mathbf{y})u(\mathbf{y}){\rm d}\mathbf{y}\right| \leq \|u\|_{L^{\infty}(K_{L})}\int_{K_{L}}|\Pi_{L}(\mathbf{x},\mathbf{y})|{\rm d}\mathbf{y}.
		$$
		There remains to prove that the last term in the above is bounded uniformly in $L$. We use
		$$
		\sum_{l=0}^{d-1} e^{2\mathrm{i}\pi l \frac{p-k}{d}}=d \mathbbm{1}_{p=k(\text{mod } d)}
		$$ 
		to compute
		\begin{align*}
		\Pi_{L}(\mathbf{x}, \mathbf{y}) & = \frac{2\sqrt{\pi}}{L^3} \sum_{k, p \in \mathbb{Z}} e^{-2\mathrm{i}\pi \frac{kx_{2}-py_{2}}{L}-\frac{1}{2}\left(x_{1}+k \frac{L}{d}\right)^{2}-\frac{1}{2}\left(y_{1}+p \frac{L}{d}\right)^{2}} \sum_{l=0}^{d-1}e^{2\mathrm{i}\pi l \frac{p-k}{d}} \\
		&= \frac{2\sqrt{\pi}d}{L^3}\sum_{k, q \in \mathbb{Z}} e^{-\frac{2\mathrm{i}\pi k}{L}(x_{2}-y_{2}) + 2\mathrm{i}\pi\frac{qdy_{2}}{L} - \frac{1}{2}\left(x_{1}+k \frac{L}{d}\right)^{2}-\frac{1}{2}\left(y_{1}+k \frac{L}{d}+qL\right)^{2}} \\
		&= \frac{1}{\sqrt{\pi}L}\sum_{k,q\in \mathbb{Z}}e^{-\frac{2\mathrm{i}\pi k}{L}(x_{2}-y_{2}) + \mathrm{i}qy_{2}-\left(\frac{2\pi k}{L}+\frac{qL+x_{1}+y_{1}}{2}\right)^{2}-\frac{1}{4}\left(qL-x_{1}+y_{1}\right)^{2}}
		\end{align*}
		Fixing $q$ and applying the Poisson summation formula in $k$, we obtain
		$$
		\Pi_{L}(\mathbf{x}, \mathbf{y}) = \frac{1}{\sqrt{\pi}} \sum_{k, q \in \mathbb{Z}} e^{\frac{\mathrm{i}}{2}(qL+x_{1}+y_{1})(kL+x_{2}-y_{2})+\mathrm{i}qy_{2}-\frac{1}{4}(kL+x_{2}-y_{2})^{2}-\frac{1}{4}\left(qL-x_{1}+y_{1}\right)^{2}}.
		$$
		By neglecting the imaginary part, we estimate
		\begin{equation}\label{appendix:bounded-1}
		|\Pi_{L}(\mathbf{x}, \mathbf{y})| \leq \frac{1}{\sqrt{\pi}} \sum_{k\in \mathbb{Z}} e^{-\frac{1}{4}(kL+x_{2}-y_{2})^{2}} \sum_{q \in \mathbb{Z}} e^{-\frac{1}{4}\left(qL-x_{1}+y_{1}\right)^{2}}.
		\end{equation}
		Note that $\mathbf{x}$ and $\mathbf{y}$ lie in the square $K_{L}$ of side length $L$ and we have
		$$
		(x_{1}-y_{1})^{2} + (x_{2}-y_{2})^{2} = |\mathbf{x} - \mathbf{y}|^{2} \leq 2L^{2}.
		$$
		Then
		\begin{align*}
		\sum_{k\in \mathbb{Z}} e^{-\frac{1}{4}(kL+x_{2}-y_{2})^{2}} & = \left(\sum_{k\in\{-1,0,1\}}+\sum_{k \geq 2}+\sum_{k \leq -2}\right) e^{-\frac{1}{4}(kL+x_{2}-y_{2})^{2}} \\
		& \leq \sum_{k\in\{-1,0,1\}}e^{-\frac{1}{4}(kL+x_{2}-y_{2})^{2}} + \sum_{k \geq 2} e^{-\frac{1}{4}(kL+x_{2}-y_{2})^{2}} + e^{-\frac{1}{4}(kL-x_{2}+y_{2})^{2}} \\
		& \leq \sum_{k\in\{-1,0,1\}}e^{-\frac{1}{4}(kL+x_{2}-y_{2})^{2}} + 2\sum_{k \geq 2} e^{-\frac{1}{4}(k-\sqrt{2})^{2}L^{2}} \\
		& \leq \sum_{k\in\{-1,0,1\}}e^{-\frac{1}{4}(kL+x_{2}-y_{2})^{2}} + Ce^{-\frac{1}{4}(2-\sqrt{2})^{2}L^{2}}.
		\end{align*}
		This implies that
		\begin{equation}\label{appendix:bounded-2}
		\int_{\left[-\frac{L}{2},\frac{L}{2}\right]} \sum_{k\in \mathbb{Z}}e^{-\frac{1}{4}(kL+x_{2}-y_{2})^{2}}{\rm d}y_{2} \leq \sum_{k\in\{-1,0,1\}} \int_{\mathbb{R}}e^{-\frac{1}{4}(kL+x_{2}-y_{2})^{2}}{\rm d}y_{2} + CLe^{-\frac{1}{4}(2-\sqrt{2})^{2}L^{2}} \leq C,
		\end{equation}
		for a universal constant $C > 0$ independent of $L$. The last inequality is obtained by a simple change of variables in the integrals. 
		
		Similarly, we have
		\begin{equation}\label{appendix:bounded-3}
		\int_{\left[-\frac{L}{2},\frac{L}{2}\right]} \sum_{q \in \mathbb{Z}} e^{-\frac{1}{4}\left(qL-x_{1}+y_{1}\right)^{2}}{\rm d}y_{1} \leq \sum_{q\in\{-1,0,1\}} \int_{\mathbb{R}}e^{-\frac{1}{4}(qL-x_{1}+y_{1})^{2}}{\rm d}y_{1} + CLe^{-\frac{1}{4}(2-\sqrt{2})^{2}L^{2}} \leq C.
		\end{equation}
		Putting together \eqref{appendix:bounded-1}, \eqref{appendix:bounded-2} and \eqref{appendix:bounded-3} we obtain
		$$
		\int_{K_{L}}|\Pi_{L}(\mathbf{x},\mathbf{y})|{\rm d}\mathbf{y} \leq C.
		$$
		This yields the desired result.
	\end{proof}
	%
	
	\section{Convergence of the GP energy to the LLL energy}\label{app:gp-lll}
	In this appendix, we study the convergence of the GP energy \eqref{energy:GP} to the LLL energy \eqref{energy:LLL} in the limit $\Omega \nearrow 1$. We will prove the following generalisation of results from \cite{AftBla-08} (which considered only the case of fixed $G$): 
	
	
	
	
	\begin{proposition}[\textbf{Reduction to the lowest Landau level}]\label{pro:LLL-infinite}~\\
		Let $G = G_{\Omega} = \left(1-\Omega^{2}\right)^{-\delta}$ with $-1 < \delta < 1$. We have, in the limit $\Omega \nearrow 1$,
		\begin{equation}\label{energy:behavior-GP-LLL}
		E_{\Omega}^{\rm GP} - E_{\Omega}^{\rm LLL} = o\left(G_{\Omega}^{\frac{1}{2}}\left(1-\Omega^{2}\right)^{\frac{1}{2}}\right).
		\end{equation} 
		Moreover, for a minimizer $\uGP$ of $\cEGP$ we have that 
		\begin{equation}\label{eq:dens GP to LLL 1}
		\lim_{\Omega \nearrow 1}\norm{\left(1-\Pi_0\right) \uGP}_{L^2} = 0 
		\end{equation}
        with $\Pi_0$ the orthogonal projector on $\LLL$ and 
        \begin{equation}\label{eq:dens GP to LLL 2}
		\cELLL \left[\frac{\Pi_0 \uGP}{\norm{\Pi_0\uGP}_{L^2}}\right] = \ELLL + o\left(\ETF \right)
		\end{equation}
		when $\Omega \nearrow 1$.
	\end{proposition}
	
	The assumption on $G$ in Proposition~\ref{pro:LLL-infinite} guarantees that $E_{\Omega}^{\rm TF} \ll 1$. Also, we have
	\begin{equation}\label{energy-lll:asymptotic}
	E_{\Omega}^{\rm LLL} = 1 + o(1) \quad \text{as} \quad \Omega \nearrow 1,
	\end{equation}
	by \eqref{energy:behavior-lll}. Part of the proof of Proposition~\ref{pro:LLL-infinite} is similar to that of \cite[Theorem 1.2]{AftBla-08}. We however need a better control when $G$ is allowed to be large in the limit $\Omega \nearrow 1$. The following elliptic estimate for the GP equation with a trapping term is our main new ingredient.
	
	\begin{lemma}[\textbf{Elliptic estimates for an inhomoegenous GP equation}]\label{lem:elliptic}~\\
		Let $u^{\rm GP}$ be a solution to 
		\begin{equation}\label{eq:EL-GP bis}
		\frac{1}{2}\big(-\mathrm{i}\nabla - \mathbf{x}^{\perp}\big)^{2}u^{\rm GP} + G_{\Omega}\big|u^{\rm GP}\big|^{2}u^{\rm GP} + \frac{1-\Omega^{2}}{2\Omega^{2}}|\mathbf{x}|^{2}u^{\rm GP} = \lambda u^{\rm GP},
		\end{equation}
		on $\mathbb{R}^{2}$, for some $\lambda \geq 0$. We have the following properties.
		\begin{itemize}
			\item[(i)] If $\lambda \leq 1$, then $u = 0$.
			\item[(ii)] There exists a universal constant $C_{\max} > 0$ such that if $\lambda > 1$, then
			\begin{equation}\label{ineq:elliptic}
			\big\|u^{\rm GP}\big\|_{L^{\infty}} \leq \left(\frac{\lambda}{G_{\Omega}}\right)^{\frac{1}{2}} \min\left\{1,C_{\max}(\lambda-1)^{\frac{1}{4}}\right\}.
			\end{equation}
		\end{itemize}
	\end{lemma}
	
	We expect that an optimal bound should be 
	$$
	\big\|u^{\rm GP}\big\|_{L^{\infty}} \leq \left(\frac{\lambda}{G_{\Omega}}\right)^{\frac{1}{2}} \min\left\{1,C_{\max}(\lambda-1)^{\frac{1}{2}}\right\}.
	$$
	This would better match similar estimates from \cite{FouHel-10}, that we take inspiration from. Moreover this would prove that the density can nowhere exceed a constant times the maximal Thomas--Fermi density, the natural scale in our problem. The above~\eqref{ineq:elliptic} is however sufficient for our purpose.
	
	\begin{proof}[Proof of Lemma~\ref{lem:elliptic}]
		It can be seen immediately from the operator inequality \eqref{ineq:operator} that Equation~\eqref{eq:EL-GP bis} admits only trivial $L^{\infty}$-solution if $\lambda \leq 1$. 
		
		Let 
		$$v^{\rm GP} = \left(\frac{G_{\Omega}}{\lambda}\right)^{\frac{1}{2}}u^{\rm GP}.$$
		Then $v^{\rm GP}$ solves the equation
		\begin{equation}\label{eq:GL-fake}
		\frac{1}{2}\big(-\mathrm{i}\nabla - \mathbf{x}^{\perp}\big)^{2}v^{\rm GP} + \frac{1-\Omega^{2}}{2\Omega^{2}}|\mathbf{x}|^{2}v^{\rm GP} = \lambda\big(1-\big|v^{\rm GP}\big|^{2}\big) v^{\rm GP}.
		\end{equation}
		In order to prove \eqref{ineq:elliptic}, we will show that
		\begin{equation}\label{ineq:elliptic-fake}
		\big\|v^{\rm GP}\big\|_{L^{\infty}} \leq \min\left\{1,C_{\max}(\lambda-1)^{\frac{1}{4}}\right\}.
		\end{equation}
		We write down the equation satisfied by $\big|v^{\rm GP}\big|^{2}$ as follows
		$$
		-\frac{1}{2}\Delta \big|v^{\rm GP}\big|^{2} + \mathrm{i}\mathbf{x}^{\perp}\left(\overline{v^{\rm GP}}\nabla v^{\rm GP} + v^{\rm GP}\overline{\nabla v^{\rm GP}}\right) + \big|\nabla v^{\rm GP}\big|^{2} + \frac{1}{\Omega^{2}}|\mathbf{x}|^{2}\big|v^{\rm GP}\big|^{2} = 2\lambda\big|v^{\rm GP}\big|^{2}\left(1-\big|v^{\rm GP}\big|^{2}\right).
		$$
		This implies
		$$
		-\frac{1}{2}\Delta \big|v^{\rm GP}\big|^{2} + \frac{1-\Omega^{2}}{\Omega^{2}}|\mathbf{x}|^{2}\big|v^{\rm GP}\big|^{2} + 2\lambda\big(\big|v^{\rm GP}\big|^{2}-1\big)\big|v^{\rm GP}\big|^{2} = -\big|\nabla v^{\rm GP} - \mathrm{i}\mathbf{x}^{\perp}\overline{v^{\rm GP}}\big|^{2} \leq 0.
		$$
		By the maximum principle, we deduce that, at any maximum point $\mathbf{x}^*$ of $\big|v^{\rm GP}\big|^{2}$,
		$$
		\big|v^{\rm GP}(\mathbf{x}^*)\big|^{2} \leq 1
		$$
		and 
		\begin{equation}\label{eq:xmax}
		\frac{1-\Omega^{2}}{\Omega^{2}}|\mathbf{x^*}|^{2} \leq 2 \lambda. 
		\end{equation} 
		There remains to prove the second half of~\eqref{ineq:elliptic-fake}, and we may assume that $\lambda \leq  2$ for that purpose.
		
		Suppose by contradiction that there exists a sequence of solutions $\{v_{n}\}$ to Equation~\eqref{eq:GL-fake} with $\lambda_{n} > 1$ and
		\begin{equation}\label{ineq:elliptic-contradiction}
		\lim_{n\to\infty}\frac{\|v_{n}\|_{L^{\infty}}}{(\lambda_{n}-1)^{\frac{1}{4}}} = \infty.
		\end{equation}
		Define $\Lambda_{n} := \|v_{n}\|_{L^{\infty}}$. Since $\Lambda_{n} \leq 1$, by the maximum principle, we must have $\lambda_{n} \to 1$ as $n\to\infty$. On the other hand, there exists a point $\mathbf{x}_{n} \in \mathbb{R}^{2}$ with $|v_{n}(\mathbf{x}_{n})| \geq \frac{\Lambda_{n}}{2}$. Consider the function $f_{n} := \Lambda_{n}^{-1}v_{n}(\cdot+\mathbf{x}_{n})$. This function satisfies
		\begin{equation}\label{ineq:maximum-principle}
		\frac{1}{2} \leq |f_{n}(0)| \leq \|f_{n}\|_{L^{\infty}} \leq 1
		\end{equation}
		and it solves the equation
		\begin{equation}\label{eq:GL-contradiction}
		\frac{1}{2}(-\mathrm{i}\nabla - (\mathbf{x}+\mathbf{x}_{n})^{\perp})^{2}f_{n} + \frac{1-\Omega^{2}}{2\Omega^{2}}|\mathbf{x}+\mathbf{x}_{n}|^{2}f_{n} = \lambda_{n}(1-\Lambda_{n}^{2}\big|f_{n}\big|^{2}) f_{n}.
		\end{equation}
		In view of~\eqref{eq:xmax} we may assume that the potential 
		$$ 
		\frac{1-\Omega^{2}}{2\Omega^{2}}|\mathbf{x}+\mathbf{x}_{n}|^{2}
		$$
		and its derivatives stay uniformly bounded if $\mathbf{x} \in B(0,L)$ with $L$ fixed. Hence, by the boundedness of $f_{n}$ in \eqref{ineq:maximum-principle} and elliptic regularity, the sequence $\{f_{n}\}$ is bounded in $W_{\rm{loc}}^{2,p}(\mathbb{R}^{2})$, for all $p < \infty$. By compactness we can find a convergent subsequence in $W_{\rm loc}^{s,p}(\mathbb{R}^{2})$, for any given $s < 2$ and $p < \infty$. Thanks to the Rellich--Kondrashov Theorem (see e.g., \cite[Theorem 8.9]{LieLos-02}) and extracting a subsequence, we find
		$$
		f_{n} \to f \in W_{\rm loc}^{s,p}(\mathbb{R}^{2}) \hookrightarrow L_{\rm loc}^{\infty}(\mathbb{R}^{2}),
		$$
		when $sp > 2$, where $f$ satisfies
		\begin{equation}\label{eq:contradict}
		\frac{1}{2} \leq |f(0)| \leq \|f\|_{L^{\infty}} \leq 1. 
		\end{equation}
		We next seek a contradiction with this finding. 
		
		Let $0\leq \chi \leq 1$ be a fixed smooth function on $\mathbb{R}^{2}$ such that $\chi(\mathbf{x}) = 1$ if $|\mathbf{x}| \leq 1$ and $\chi(\mathbf{x}) = 0$ if $|\mathbf{x}| \geq 2$. For $R>0$, we denote $\chi_{R} = \chi\left(\frac{\cdot}{R}\right)$. Note that it is possible to construct the function $\chi$ so as to satisfy
		$$
		|\nabla \chi_{R}| \leq CR^{-1}\chi_{R}^{1-\mu},
		$$
		for some arbitrarily small $\mu > 0$, independent of $R$, e.g., by taking $\chi = h^{\nu}$ for $\nu$ large and some smooth function $0\leq h \leq 1$. Now, we use the identity
		$$
		-\Re\left\langle f_{n},\chi_{R}^{2}\Delta f_{n} \right\rangle = \int_{\mathbb{R}^{2}}\big|\nabla \chi_{R}f_{n}\big|^{2} - \int_{\mathbb{R}^{2}}|\nabla \chi_{R}|^{2}\big|f_{n}\big|^{2}.
		$$
		Then, we integrate the equation \eqref{eq:GL-contradiction} against $\chi_{R}^{2}\overline{f_{n}}$ and drop the quadratic term. We obtain
		\begin{equation}\label{ineq:GL-contradiction}
		\frac{1}{2}\int_{\mathbb{R}^{2}}\big|\big(-\mathrm{i}\nabla - (\mathbf{x}+\mathbf{x}_{n})^{\perp}\big)\chi_{R}f_{n}\big|^{2} + \int_{\mathbb{R}^{2}}\lambda_{n}\Lambda_{n}^{2}\chi_{R}^{2}\big|f_{n}\big|^{4} \leq \lambda_{n}\int_{\mathbb{R}^{2}}\chi_{R}^{2}\big|f_{n}\big|^{2} + \frac{C}{R^{2}}\int_{R \leq |\mathbf{x}| \leq 2R}\chi_{R}^{2-2\mu}\big|f_{n}\big|^{2}.
		\end{equation}
		Note that
		\begin{equation}\label{ineq:GL-contradiction-1}
		\frac{1}{2}\int_{\mathbb{R}^{2}} \big|\big(-\mathrm{i}\nabla - (\mathbf{x}+\mathbf{x}_{n})^{\perp}\big)\chi_{R}f_{n}\big|^{2} = \frac{1}{2}\int_{\mathbb{R}^{2}} \big|\big(-\mathrm{i}\nabla - \mathbf{x}^{\perp}\big)e^{\mathrm{i}\phi_{n}}\chi_{R}f_{n}\big|^{2} \geq \int_{\mathbb{R}^{2}}\chi_{R}^{2}\big|f_{n}\big|^{2},
		\end{equation}
		where the phase $\phi_{n}$ is chosen in such a way that
		$$
		\mathbf{x}^{\perp} - (\mathbf{x}+\mathbf{x}_{n})^{\perp} = \nabla\phi_{n}.
		$$
		On the other hand, by H\"older' inequality, we have
		\begin{equation}\label{ineq:GL-contradiction-2}
		\int_{R \leq |\mathbf{x}| \leq 2R}\chi_{R}^{2-2\mu}\big|f_{n}\big|^{2} \leq C\left(\int_{R \leq |\mathbf{x}| \leq 2R}\chi_{R}^{2-4\mu}\right)^{\frac{1}{2}}\left(\int_{\mathbb{R}^{2}} \chi_{R}^{2}\big|f_{n}\big|^{4}\right)^{\frac{1}{2}} \leq CR\left(\int_{\mathbb{R}^{2}} \chi_{R}^{2}\big|f_{n}\big|^{4}\right)^{\frac{1}{2}}.
		\end{equation}
		Putting together \eqref{ineq:GL-contradiction}, \eqref{ineq:GL-contradiction-1}, \eqref{ineq:GL-contradiction-2} and choosing $R = R_{n} = (\lambda_{n}-1)^{-\frac{1}{2}} \to \infty$ we obtain
		$$
		\int_{\mathbb{R}^{2}} \chi_{R_{n}}^{2}\big|f_{n}\big|^{4} \leq C\frac{\lambda_{n}-1}{\Lambda_{n}^{4}}.
		$$
		Then, \eqref{ineq:elliptic-contradiction} implies that
		$$
		\lim_{n\to\infty}\int_{\mathbb{R}^{2}} \chi_{L}^{2}\big|f_{n}\big|^{4} = 0, \quad \forall L>0.
		$$
		The convergence $f_{n} \to f$ in $W_{\rm loc}^{s,p}(\mathbb{R}^{2})$ yields that $f=0$. This contradicts~\eqref{eq:contradict} and shows that we must have \eqref{ineq:elliptic-fake}.
	\end{proof}
	
	Now we may conclude the proof of Proposition~\ref{pro:LLL-infinite}.
	
	\begin{proof}[Proof of Proposition~\ref{pro:LLL-infinite}]
		By definition, we clearly have
		$$
		E_{\Omega}^{\rm GP} \leq E_{\Omega}^{\rm LLL}.
		$$
		In order to bound $E_{\Omega}^{\rm GP}$ from below, we denote by $u^{\rm GP}$ one of its minimizers. We decompose $u^{\rm GP}$ as follows
		$$
		u^{\rm GP} = \Pi_{0}u^{\rm GP} + \Pi_{0}^{\perp}u^{\rm GP},
		$$
		where $\Pi_{0}$ is the projection on the lowest Landau level $\mathcal{LLL}$ in \eqref{space:LLL}. The kernel of $\Pi_{0}$ is given explicitly by (see e.g., \cite{LieSolYng-94b})
		$$
		\Pi_{0}(\mathbf{x},\mathbf{y}) = \frac{1}{2\pi}e^{\frac{\mathbf{i}}{2}(x_{1}y_{2}-x_{2}y_{1})}e^{-\frac{1}{2}|\mathbf{x}-\mathbf{y}|^{2}}.
		$$
		We recall that $u^{\rm GP}$ solves the equation \eqref{eq:EL-GP bis}. By \eqref{ineq:elliptic} and the upper bound on $E_{\Omega}^{\rm GP}$ in \eqref{upper:local-density-approximation} we have
		\begin{equation}\label{energy-lower:behavior-GP-LLL-error-1}
		\big\|\Pi_{0}u^{\rm GP}\big\|_{L^{\infty}} \leq \big\|u^{\rm GP}\big\|_{L^{\infty}} \leq C\left(\frac{\lambda}{G_{\Omega}}\right)^{\frac{1}{2}} (\lambda-1)^{\frac{1}{4}} \leq CG_{\Omega}^{-\frac{1}{2}}\left(E_{\Omega}^{\rm TF}\right)^{\frac{1}{4}}.
		\end{equation}
		Here we have used the fact that $\Pi_{0}$ is a bounded operator on $L^{2}\cap L^{\infty}(\mathbb{R}^{2})$. Furthermore, by \eqref{ineq:operator} and \eqref{upper:local-density-approximation}, we have
		\begin{equation}\label{energy-lower:behavior-GP-LLL-error-2}
		\int_{\mathbb{R}^{2}}\big|\Pi_{0}u^{\rm GP}\big|^{4} \leq \int_{\mathbb{R}^{2}}\big|u^{\rm GP}\big|^{4} \leq G_{\Omega}^{-1}E_{\Omega}^{\rm TF}.
		\end{equation}
		On the other hand, it is well-known that the second eigenvalue of $\frac{1}{2}\big(-\mathrm{i}\nabla - \mathbf{x}^{\perp}\big)^{2}$ is $3$ (see e.g., \cite{RouYng-19}). Consequently,
		\begin{equation}\label{ineq:second-eigenvalue}
		\frac{1}{2}\int_{\mathbb{R}^{2}} \big|\big(-\mathrm{i}\nabla - \mathbf{x}^{\perp}\big)\Pi_{0}^{\perp}u^{\rm GP}\big|^{2} \geq 3\int_{\mathbb{R}^{2}} \big|\Pi_{0}^{\perp}u^{\rm GP}\big|^{2}.
		\end{equation}
		Therefore,
		\begin{align*}
		E_{\Omega}^{\rm GP} = \mathcal{E}_{\Omega}^{\rm GP}\big[u^{\rm GP}\big] & \geq \frac{1}{2}\int_{\mathbb{R}^{2}} \big|\big(-\mathrm{i}\nabla - \mathbf{x}^{\perp}\big)\Pi_{0}u^{\rm GP}\big|^{2} + \big|\big(-\mathrm{i}\nabla - \mathbf{x}^{\perp}\big)\Pi_{0}^{\perp}u^{\rm GP}\big|^{2} \\
		& \geq \int_{\mathbb{R}^{2}} \big|\Pi_{0}u^{\rm GP}\big|^{2} + 3\big|\Pi_{0}^{\perp}u^{\rm GP}\big|^{2} \\
		& = 1 + 2\int_{\mathbb{R}^{2}} \big|\Pi_{0}^{\perp}u^{\rm GP}\big|^{2}.
		\end{align*}
		Then \eqref{upper:local-density-approximation} implies that
		\begin{equation}\label{energy-lower:projection}
		\int_{\mathbb{R}^{2}} \big|\Pi_{0}^{\perp}u^{\rm GP}\big|^{2} \leq \frac{1}{2}E_{\Omega}^{\rm TF},
		\end{equation}
		which proves~\eqref{eq:dens GP to LLL 1}. Expanding the quartic term of the energy as in \eqref{scaling-law:lower-bound} (see also \cite{AftBla-08}), we obtain
		\begin{equation}\label{energy-lower:behavior-GP-LLL}
		\mathcal{E}_{\Omega}^{\rm GP}\big[u^{\rm GP}\big] \geq \mathcal{E}_{\Omega}^{\rm GP}\big[\Pi_{0}u^{\rm GP}\big] + \mathcal{E}_{\Omega}^{\rm GP}\big[\Pi_{0}^{\perp}u^{\rm GP}\big] - \frac{G_{\Omega}}{2}\int_{\mathbb{R}^{2}}\big|\Pi_{0}^{\perp}u^{\rm GP}\big|^{4} - 2G_{\Omega}\int_{\mathbb{R}^{2}}\big|\Pi_{0}u^{\rm GP}\big|^{3}\big|\Pi_{0}^{\perp}u^{\rm GP}\big|. 
		\end{equation}
		For the main term in \eqref{energy-lower:behavior-GP-LLL}, we have
		\begin{align}\label{energy-lower:behavior-GP-LLL-1}
		\mathcal{E}_{\Omega}^{\rm GP}\big[\Pi_{0}u^{\rm GP}\big] \geq \big\|\Pi_{0}u^{\rm GP}\big\|_{L^{2}}^{4}\mathcal{E}_{\Omega}^{\rm GP}\left[\frac{\Pi_{0}u^{\rm GP}}{\big\|\Pi_{0}u^{\rm GP}\big\|_{L^{2}}}\right] & \geq \left(1-2\big\|\Pi_{0}^{\perp}u^{\rm GP}\big\|_{L^{2}}^{2}\right)E_{\Omega}^{\rm LLL} \nonumber\\
		& \geq E_{\Omega}^{\rm LLL} - 3\big\|\Pi_{0}^{\perp}u^{\rm GP}\big\|_{L^{2}}^{2}.
		\end{align}
		Here we have used \eqref{energy-lll:asymptotic}. The first error term in \eqref{energy-lower:behavior-GP-LLL} is estimated simply, by \eqref{ineq:second-eigenvalue}, as follows
		\begin{equation}\label{energy-lower:behavior-GP-LLL-2}
		\mathcal{E}_{\Omega}^{\rm GP}\big[\Pi_{0}^{\perp}u^{\rm GP}\big] - \frac{G_{\Omega}}{2}\int_{\mathbb{R}^{2}}\big|\Pi_{0}^{\perp}u^{\rm GP}\big|^{4} \geq 3\big\|\Pi_{0}^{\perp}u^{\rm GP}\big\|_{L^{2}}^{2}.
		\end{equation}
		Finally, by \eqref{energy-lower:behavior-GP-LLL-error-1}, \eqref{energy-lower:behavior-GP-LLL-error-2} and H\"older' inequality we have
		\begin{equation}\label{energy-lower:behavior-GP-LLL-3}
		2G_{\Omega}\int_{\mathbb{R}^{2}}\big|\Pi_{0}u^{\rm GP}\big|^{3}\big|\Pi_{0}^{\perp}u^{\rm GP}\big| \leq 2G_{\Omega}\big\|\Pi_{0}u^{\rm GP}\big\|_{L^{\infty}}\big\|\Pi_{0}u^{\rm GP}\big\|_{L^{4}}^{2}\big\|\Pi_{0}^{\perp}u^{\rm GP}\big\|_{L^{2}}^{2} \leq C\left(E_{\Omega}^{\rm TF}\right)^{\frac{5}{4}}.
		\end{equation}
		The error term in the above is of order $o\left(E_{\Omega}^{\rm TF}\right)$. This is the place where the condition on $G$ in Proposition~\ref{pro:LLL-infinite} is used. Putting together \eqref{energy-lower:behavior-GP-LLL}--\eqref{energy-lower:behavior-GP-LLL-3} we obtain the desired result \eqref{energy:behavior-GP-LLL}. Returning to~\eqref{energy-lower:behavior-GP-LLL} and \eqref{energy-lower:behavior-GP-LLL-1}, inserting the estimates just found and the upper bound $\EGP \leq \ELLL$ we deduce~\eqref{eq:dens GP to LLL 2}.
	\end{proof}
	

%
%
\end{document}